\documentclass[11pt]{article}
\usepackage{graphicx} 
\usepackage[margin=1in]{geometry}
\usepackage{mathpazo}
\usepackage{amsmath}
\usepackage{xcolor}
\usepackage{amsthm}
\usepackage{amssymb}
\usepackage[backref=page]{hyperref}
\hypersetup{colorlinks=true,linkcolor=blue,citecolor=[rgb]{0.1,0.1,0.9}}
\usepackage{cleveref}
\usepackage{enumerate}
\usepackage{caption}
\usepackage{subcaption}
\usepackage{thm-restate}
\usepackage{thmtools}
\usepackage{nicefrac}
\DeclareMathOperator*{\E}{\mathbb{E}}

\usepackage{setspace}
\newtheorem{theorem}{Theorem}[section]

\newtheorem{corollary}[theorem]{Corollary}
\newtheorem{lemma}[theorem]{Lemma}

\newtheorem{claim}[theorem]{Claim}
\newtheorem{remark}[theorem]{Remark}

\newtheorem{ques}[theorem]{Open Question}
\newtheorem{obs}[theorem]{Observation}
\newtheorem{defi}[theorem]{Definition}

\newcommand{\ed}{\Delta_{\textnormal{\textsf{edit}}}}
\newcommand{\selfed}{\Delta_{\widetilde{\textnormal{\textsf{edit}}}}}
\newcommand{\ham}{\Delta_{\textnormal{\textsf{Hamming}}}}
\newcommand{\lcs}{\textnormal{\textsf{LCS}}}
\newcommand{\slcs}{\widetilde{\textnormal{\textsf{LCS}}}}
\newcommand{\cost}{\textnormal{\textsf{cost}}}
\newcommand{\sigmain}{\Sigma_{\textnormal{in}}}
\newcommand{\sigmaout}{\Sigma_{\textnormal{out}}}

\renewcommand{\tilde}{\widetilde}

\title{Constant Rate Isometric Embeddings of\\ Hamming Metric into Edit Metric}
\author{Sudatta Bhattacharya\footnote{Computer Science Institute of Charles University, Prague. \url{sudatta@iuuk.mff.cuni.cz}}\and Sanjana Dey\footnote{UMONS – Université de Mons. \url{info4.sanjana@gmail.com}}\and Elazar Goldenberg\footnote{The Academic College of Tel Aviv-Yaffo. \url{elazargo@mta.ac.il}}\and Mursalin Habib\footnote{Rutgers University. \url{mursalin.habib@rutgers.edu}}\\\and Bernhard Haeupler\footnote{INSAIT, Sofia University "St. Kliment Ohridski" \& ETH Zürich. \url{bernhard.haeupler@insait.ai}} \and Karthik C.\ S.\footnote{Rutgers University. \url{karthik.cs@rutgers.edu}}\and Michal Kouck\'y\footnote{Computer Science Institute of Charles University, Prague. \url{koucky@iuuk.mff.cuni.cz}}}

\date{}

\begin{document}

\maketitle
\begin{abstract}
    A function $\varphi: \{0,1\}^n \to \{0,1\}^N$ is called an \textit{isometric embedding} of the $n$-dimensional Hamming metric space to the $N$-dimensional edit metric space if, for all $x, y \in \{0,1\}^n$, the Hamming distance between $x$ and $y$ is equal to the edit distance between $\varphi(x)$ and $\varphi(y)$. The \textit{rate} of such an embedding is defined as the ratio $n/N$.\vspace{0.1cm}

It is well known in the literature how to construct isometric embeddings with a rate of $\Omega\left(\frac{1}{\log n}\right)$. However, achieving even near-isometric embeddings with a positive constant rate has remained elusive until now.\vspace{0.1cm}

In this paper, we present an isometric embedding with a rate of $\frac{1}{8}$ by discovering connections to \emph{synchronization strings}, which were studied in the context of insertion-deletion codes (Haeupler-Shahrasbi~[JACM'21]).
At a technical level, we introduce a framework for obtaining high-rate isometric embeddings using a novel object called a  \emph{misaligner}. We speculate that, with sufficient computational resources, our framework could potentially yield isometric embeddings with a rate of $\frac{1}{5}$.\vspace{0.1cm}

As an immediate consequence of our constant rate isometric embedding, we improve known conditional lower bounds for the closest pair problem and the discrete 1-center problem in the edit metric and NP-hardness of approximation results for clustering problems and the Steiner tree problem in the edit metric, but now with optimal dependency on the dimension. \mbox{Furthermore}, we obtain optimal lower bounds for the gap edit distance problem in the two-player randomized communication complexity model.\vspace{0.1cm}

We complement our results by showing that no isometric embedding $\varphi:\{0, 1\}^n \to \{0, 1\}^N$ can have rate greater than $\frac{15}{32}$ for all positive integers $n$. En route to proving this upper bound, we uncover fundamental structural properties necessary for every Hamming-to-edit isometric embedding. We also prove similar upper and lower bounds for embeddings over larger alphabets. \vspace{0.1cm}

Finally, we consider embeddings $\varphi:\Sigma_{\text{in}}^n\to \Sigma_{\text{out}}^N$ between different input and output alphabets, where the rate is given by $\frac{n\log|\Sigma_{\text{in}}|}{N\log|\Sigma_{\text{out}}|}$. In this setting, we show that the rate can be made arbitrarily close to 1.

\end{abstract} 
\thispagestyle{empty}
\clearpage

\setcounter{page}{1}

\section{Introduction}

Metric embeddings offer a powerful framework for formally comparing metric spaces by mapping points from a source space into a target space. The goal of such a mapping is to preserve pairwise distances faithfully (i.e.,  minimize \textit{distortion}), thereby revealing structural and computational relationships between different notions of distance. These embeddings are particularly useful for understanding connections among fundamental metrics, such as those central to coding theory, sequence analysis, and computational geometry.

Among the most fundamental distance measures arising in stringology and related fields are the Hamming and the edit distances. The edit distance quantifies the minimum number of character insertions, deletions, and substitutions required to transform one string into another, while restricting these operations to only substitutions on equal-length strings yields the Hamming distance. Because these metrics capture different notions of string similarity, both play central roles in diverse areas, including pattern matching, machine learning, and computational biology \cite{navarro2001guided,liljas2016textbook}. Consequently, numerous classical computational problems -- like finding the closest pair of strings in a set \cite{marzal1993computation,gao2010survey} or identifying a representative center string for a dataset \cite{WF74, NR05, LMW02} -- are widely studied using these metrics, underpinning even public bioinformatics services \cite{NCB}.

Given the fundamental nature of both metrics, understanding their relationship via embeddings is crucial. This paper systematically studies embeddings of the Hamming metric into the edit metric, exploring how the simpler substitution-only distance structure can be represented within the richer, insertion- and deletion-allowing space.   Formally, we investigate functions $\varphi: \{0, 1\}^n \to \{0, 1\}^N$, where the domain $\{0, 1\}^n$ is equipped with the Hamming distance $\ham$ and the codomain $\{0, 1\}^N$ with the edit distance $\ed$. We define the rate of the embedding $\varphi$ to be $\frac{n}{N}$. If $\varphi$ keeps distances unchanged except possibly multiplying all of them by some fixed constant, i.e., for some universal constant $K\ge 1$, we have  $\ed~(\varphi(x),~\varphi(y))=K~\cdot~\ham(x, y) $, for all $x,  y \in \{0, 1\}^n$, then we call $\varphi$ a $1$-embedding. If $K=1$, i.e., $\varphi$ preserves distances exactly, then we call $\varphi$ an \textit{isometric embedding}.

The goal of this work is to seek answers to fundamental questions about these embeddings: Is it possible to achieve a positive constant rate 1-embedding, or better, an isometric embedding? What is the optimal rate achievable? Can we establish bounds on rates that are unattainable?   Does using larger alphabets for isometric embeddings unlock the ability to obtain higher rates? By addressing all these questions, we aim to provide a comprehensive study of embeddings of the Hamming metric into the edit metric.

\subsection{Quest for Constant Rate 1-Embedding}
\label{sec:quest-for-constant}
In computational complexity, it is typically easier to prove the intractability of problems in the Hamming metric, and in such a situation, if we had a 1-embedding of the Hamming metric into the edit metric, then the hardness result would translate to the edit metric as well. These embeddings have been studied in the literature for exactly this purpose. 
For instance, a folklore embedding\footnote{See, for example, Lemma 11 in \cite{ABCSS23}.} from the Hamming metric to the edit metric constructs the output string by inserting a uniformly random block of size $\Theta(\log n)$ after each character of the input string. This can be shown to result in an isometric embedding with high probability (and can be derandomized using constant relative distance codes in the edit metric, such as the ones in~\cite{schulman1999asymptotically}).

However, this embedding achieves a rate of only $\frac{1}{\Theta(\log n)}$. Consequently, this implies that if a problem is hard for strings of length $n$ in the Hamming metric, then it remains hard in the edit metric, but only for strings of length $\Theta(n \log n)$. 
On the other hand, our intuition suggests that solving problems in the $n$-dimensional edit metric should be at least as hard as solving the same problems in the $n$-dimensional Hamming metric.
The inability to provide a formal justification for this intuition leaves us wanting for a clearer understanding. To address this gap, we aim to answer the following natural question:

\begin{center}
    \textit{Does there exist a 1-embedding of the Hamming metric into the edit metric\\ with positive constant rate?}
\end{center}

Our first result is an affirmative answer to the above question (in fact, we provide an isometric embedding). Our conceptual contribution in this regard is a clean connection between \emph{synchronization strings} studied in the context of the study of codes for correcting insertion and deletion errors~\cite{haeupler2017synchronization, haeupler2017synchronization2, haeupler2017synchronization3, haeupler2018synchronization4, cheng2018synchronization} with the above question. 

\begin{theorem}\label{thm:const}
There exists a universal constant $C\ge 1$ such that for every positive integer $n$, there is an isometric embedding $\varphi_{n}:\{0,1\}^n\to\{0,1\}^{Cn}$ of the Hamming metric into the edit metric.
\end{theorem}

 This represents a sharp improvement: even for the weaker notion of \emph{near-isometric} embeddings, the previously best known rate was only $\frac{1}{\Theta(\log n)}$~\cite{Rubinstein18}.

 We remark here that although the proof of Theorem~\ref{thm:const} is, in hindsight, simple and natural (see Section~\ref{sec:sync-proof-overview} for a proof sketch), the route to it was not obvious. Synchronization strings have been  designed in literature to trade alphabet size for approximation accuracy, and had been regarded primarily as an approximation tool. It was not evident, at least to us, that they could be leveraged to obtain isometric embeddings.

As an immediate consequence of Theorem~\ref{thm:const}, we obtain the following meta theorem for \emph{discrete} optimization problems in the edit metric. By discrete, we mean that the input is a set of points in the metric space, and the solution is some subset of the input points. For such problems, we have the following.

\begin{center}\textit{
If a discrete optimization problem defined in the $n$-dimensional Hamming metric\\ cannot be solved in time $T(n)$ (for some computable function $T$),\\ then the same optimization problem defined in the $n$-dimensional Edit metric\\ cannot be solved in time $T(\Theta(n))$.}
\end{center}

As concrete manifestations of this meta theorem, we prove optimal hardness results in two distinct settings. We remark that in the case of all the three theorems below, i.e., Theorems~\ref{thm:BCP}, \ref{thm:1-cen}, and \ref{thm:NP}, the previously known lower bounds or hardness results in the edit metric over binary strings was only for dimensions $ d = O(\log n \cdot \log\log n) $, where the earlier mentioned embedding with rate $ \Theta\left(\frac{1}{\log d}\right) $ was used. Thus, by reducing the dimension to $d=O(\log n)$ in the three theorems below, we obtain optimal dependency in the dimension.

\paragraph{Fine-Grained Complexity.}
Using Theorem~\ref{thm:const}, we obtain conditional lower bounds for the closest pair problem and the 1-center problem in the edit metric with optimal dependency in the dimension.

\begin{restatable}[]{theorem}{BCP}
\label{thm:BCP}
Unless the Strong Exponential Time Hypothesis
is false, 
for every $\delta>0$ there exists an $\varepsilon>0$ such that given as input sets
$A, B \subseteq  \{0, 1\}^d$
of $N$ vectors (where $d = O_{\delta} (\log N)$), computing a
$(1 + \varepsilon)$-approximate closest pair in $A\times B$ in the edit metric requires time
$\Omega(N^{2-\delta})$.
\end{restatable}

The above theorem follows immediately by combining Theorem~\ref{thm:const} with the conditional lower bound for the bichromatic closest pair problem in the Hamming metric obtained by \cite{Rubinstein18}. Moreover, one can also obtain a conditional lower bound of $n^{1.5-o(1)}$ (with optimal dependency on dimension) against approximating the monochromatic closest pair problem in the edit metric by starting from~\cite{karthik2020closest}. Next, for the discrete 1-center problem, we have the following:

\begin{restatable}[]{theorem}{1-center}
\label{thm:1-cen}
Unless the Hitting Set Conjecture is false, for every $\varepsilon>0$ there exists $c>1$ such that given as input a point-set $P\subseteq \{0,1\}^d$ where $|P|=N$ and $d=c\log N$, computing the point $x \in P$ that minimizes the maximum edit distance to all points in $P$ requires $\Omega(N^{2-\varepsilon})$ time.
\end{restatable}

The above theorem follows immediately by combining Theorem~\ref{thm:const} with the conditional lower bound for the discrete 1-center problem in the Hamming metric obtained by \cite{ABCSS23}.

We also obtain improved lower bounds for data structures supporting dictionary look-ups and text indexing under the edit distance.

\begin{theorem}
\label{thm:dictionary}
Assuming the Strong Exponential Time Hypothesis, for every $\delta > 0$, there exists a constant $c'$ such that any data structure with polynomial construction time that solves the dictionary look-up problem under the edit distance for a set of $n$ binary strings of length $c' \log n$ cannot have query time $O(n^{1 - \delta})$.
\end{theorem}

\begin{theorem}
\label{thm:indexing}
Assuming the Strong Exponential Hypothesis, for every $\delta > 0$, there exists a constant \(c'\) such that any data structure for text indexing under the edit distance for a text of length $n$, which can be constructed in polynomial time, cannot have query time $O(n^{1 - \delta})$ even when pattern strings have length at most \(c' \log n\).
\end{theorem}

Theorems~\ref{thm:dictionary} and~\ref{thm:indexing} follow directly by applying the embedding from Theorem~\ref{thm:const} to the proofs of Corollaries 7 and 19 in~\cite{cohen2019lower}, respectively.

\paragraph{NP Hardness.} Applying Theorem~\ref{thm:const} to known NP-hardness of approximation results for discrete clustering problems \cite{CKL22} and the discrete Steiner tree problem \cite{FGK24} in the Hamming metric, we obtain below their hardness of approximation result in the edit metric with optimal dependency in the dimension.

\begin{theorem}\label{thm:NP}
It is NP-hard to approximate:
\begin{itemize}
    \item the discrete $k$-means problem (resp.\ discrete $k$-center problem and discrete $k$-median problem) on $N$ points in the $\{0,1\}^{O(\log N)}$ dimensional edit metric to a factor better than $1.38$ (resp.\ $3-o(1)$ and $1.12$).
    \item the discrete Steiner tree problem on $N$ terminals in the $\{0,1\}^{O(\log N)}$ dimensional edit metric to a factor better than $1.004$.
\end{itemize}
    \end{theorem}

Beyond the meta-theorem for discrete optimization, the framework of transferring hardness via isometric embeddings extends to other computational models and complexity measures, including communication complexity, streaming algorithms, distributed algorithms, and online algorithms. As an illustration, we now discuss an application to communication complexity.

\paragraph{Communication Complexity.} In the \emph{Gap-Hamming problem}, Alice and Bob are each given $n$-bit strings, and the goal is to design a communication protocol that allows one party (say, Alice) to compute the Hamming distance between their strings to within $\pm \sqrt{n}$ using as little communication as possible. This problem was introduced by \cite{IndykW03}, and its complexity was settled by \cite{ChakrabartiR12}, who showed an $\Omega(n)$ randomized communication complexity lower bound (see also \cite{vidick2012concentration, sherstov2012communication, hadar2019communication}). The Gap-Hamming problem is of particular interest in communication complexity as it has applications to proving lower bounds for many streaming algorithms, for example, frequency moment estimation \cite{IndykW05} and entropy estimation \cite{ChakrabartiCM10}.

We introduce a natural edit distance analogue: the \emph{Gap-Edit problem} where Alice and Bob receive $n$-bit strings and must estimate their edit distance within $\pm\sqrt{n}$ using as little communication as possible. Applying Theorem~\ref{thm:const} to the $\Omega(n)$ lower bound for the Gap-Hamming problem \cite{ChakrabartiR12} immediately implies an identical $\Omega(n)$ randomized communication complexity lower bound for the Gap-Edit problem.  To the best of our knowledge, the Gap-Edit problem under this additive approximation has not been studied in the literature; a closely related multiplicative-approximation version has been investigated in \cite{AndoniK10}.

\subsection{High-Rate Isometric Embeddings}

While Theorem~\ref{thm:const} provides a constant-rate isometric embedding of the Hamming metric into the edit metric, the guaranteed rate is quite small. In certain natural applications, however, it is important that we have \textit{high-rate} embeddings. Below, we present three concrete motivations for finding isometric embeddings that maximize the rate.

\paragraph{Motivation 1: The Exact-Edit Problem.} Consider the following special case of the above mentioned Gap-Hamming problem (resp.\ Gap-Edit problem), named the \emph{Exact-Hamming problem} (resp.\ \emph{Exact-Edit problem}), where Alice and Bob are each given $n$-bit strings, and the goal is to design a communication protocol that lets one party compute the Hamming distance (resp.\ edit distance) between their strings \emph{exactly}. The $\Omega(n)$ randomized communication complexity for the Gap-Hamming problem can be improved to $(1-o_{\varepsilon}(1))\cdot n$ for the Exact-Hamming problem\footnote{This improved lower bound is obtained by a reduction from the Inner Product problem, where Alice and Bob must compute the parity of the number of coordinates where both have a $1$. The randomized communication complexity for the Inner Product problem with error probability $\frac{1}{2}-\varepsilon$ is at least $n-O(\log(1/\varepsilon))$ \cite{CG88,KN96}. Note that for every $x,y\in \{0,1\}^n$, we have $\sum_{i\in [n]}x_i\cdot y_i=(|x|+|y|-\ham(x,y))/2$, and thus every protocol for the Exact-Hamming problem with $t$ bits of communication implies a protocol for the Inner Product problem with $t+O(\log n)$ bits of communication.}, where $\varepsilon$ is the error probability. This improvement is achieved by explicitly using the structure of the Hamming metric, and it is not clear how to extend it to the Exact-Edit problem. Directly applying Theorem~\ref{thm:const} only yields that the randomized communication complexity for the Exact-Edit problem is $\Omega(n)$. Thus, to obtain a stronger lower bound, one seeks a high-rate isometric embedding.

\paragraph{Motivation 2: Codes in Edit Metric.}
Theorem~\ref{thm:const} directly implies that good Hamming codes (those with positive constant rate and relative distance) yield corresponding codes in the edit metric. Since the isometric embedding given by the theorem simply involves interleaving the input bits with a sequence of fixed bit strings\footnote{Such embeddings are referred to as interleaved embeddings throughout the paper; see Section~\ref{sec:intro-padding}.}, these embedded codes preserve many structural properties of their Hamming counterparts. The drawback, however, is that the embedding reduces the code's rate and relative distance by the constant factor $C$ introduced in the theorem statement. 

\paragraph{Motivation 3: Hamming Cube in Edit Cube.} 
Consider the weighted complete graph $\tilde{Q}_n$ (resp.\ $\tilde{A}_N$) on the vertex set $\{0,1\}^n$ (resp.\ $\{0,1\}^N$) where the weight of the edge on any two vertices corresponds to their Hamming distance (resp.\ edit distance). 
It is clear that for every pair of vertices, their weight in $\tilde A_n$ is at most their weight in $\tilde Q_n$. 
A natural question then arises: what is the smallest $N> n$, such that we can identify an isomorphic copy of $\tilde Q_n$ in $\tilde A_N$? This question of finding the largest Hamming cube that embeds isometrically into an edit cube is of pure combinatorial interest, independent of other applications.

The above three motivations drive us to find high-rate isometric embeddings of 
the Hamming metric into the edit metric\footnote{For the first two motivating questions, it would suffice to have a high-rate 1-embedding instead of an isometric embedding. However, we speculate that any 1-embedding with rate more than $1/3$ needs to be isometric.  See Remark~\ref{rem:1embed} for details.
} leading to the following fundamental question:

\begin{center}
    \textit{What is the optimal rate for an isometric embedding\\ of the Hamming metric into the Edit metric?}
\end{center}

We make significant progress on this question and prove the following:

\begin{theorem}\label{thm:explicit}
There is an isometric embedding of rate $\tfrac{1}{8}$, i.e., for every positive integer $n$, there is an isometric embedding $\varphi_{n}:\{0,1\}^n\to\{0,1\}^{8n}$ of the Hamming metric into the edit metric.\end{theorem}

\begin{sloppypar}
Consequently, we obtain a randomized communication complexity lower bound of \mbox{\((\tfrac{1}{8}-o_{\varepsilon}(1))n\)} for the Exact-Edit problem. Determining the optimal leading constant is an interesting open problem; in contrast to the  {Exact-Hamming} problem, we suspect the randomized communication complexity of the {Exact-Edit} problem might be \((1-\delta)n\) bits for some constant \(\delta>0\).

Moreover, Theorem~\ref{thm:explicit} gives, in a black-box way, edit-metric codes of positive constant rate with relative distance approaching \(\tfrac{1}{16}\). Explicit constructions with larger relative distance are known, but they all rely on tailored analyses that exploit the structure of the codewords. 
Finally, a combinatorial consequence of Theorem~\ref{thm:explicit} is the realization of the $\tfrac{n}{8}$-dimensional Hamming cube in the $n$-dimensional edit cube.
\end{sloppypar}

To prove Theorem~\ref{thm:explicit}, we introduce objects called \emph{misaligners}, which can be thought of as codes in the edit metric with robust distance guarantees, in two distinct ways. First, misaligner codewords contain \textit{wildcard symbols} and,  irrespective of how these wildcards are instantiated, a large pairwise distance is guaranteed between any two distinct codewords. 
 The second guarantee is even stronger: not only is the distance between pairs of codewords large, but the distance also remains large even when comparing a codeword with a concatenation of multiple codewords. This protects against situations where the concatenation of the suffix and prefix of two codewords might closely resemble another codeword. Formally defining misaligners requires a bit of work, but below, when we speak of $(\alpha,m)$-misaligners, we informally refer to misaligners that have codewords of length $m$ and relative distance $\alpha$. We elaborate more on the parameters of misaligners in Section~\ref{sec:overview-2}, and it is formally defined in Definition~\ref{def:misaligners}.

In addition to misaligners, we also explicitly formulate objects called  $\varepsilon$-\textit{locally self-matching strings}, which are strings in which every substring of length $\ell$ has non-vertical matches of size at most $\varepsilon\ell$, where, by non-vertical matches, we forbid matching an index to itself. These objects appear implicitly in the context of synchronization strings introduced by Haeupler and Shahrasbi \cite{haeupler2017synchronization}. Again, we elaborate more on this notion in Section~\ref{sec:overview-2}, and it is formally defined in Definition~\ref{def: self-matching}.

One of our main technical contributions is showing how to utilize the construction of misaligners and locally self-matching strings to design high-rate isometric embeddings of the Hamming metric into the edit metric.

\begin{theorem}[Informal statement of Theorem~\ref{thm:main}]\label{thm:informalmain}
 Suppose for some $\varepsilon,\alpha>0$ and some integer $m$, there exists an $(\alpha,m)$-misaligner and a $\varepsilon$-{locally self-matching string}, then for every positive integer $n$ there is an isometric embedding $\varphi_{n}:\{0,1\}^n\to\{0,1\}^{Rn}$ of the Hamming metric into the edit metric, where:
 $$R= \left\lceil\frac{1}{(1-\varepsilon)\alpha-\frac{1}{3m-1}}\right\rceil.$$
\end{theorem}

Even with limited computing resources, we were able to construct a $(0.1625, 320)$-misaligner and a $0.224$-locally self-matching string. This yields the embedding in Theorem~\ref{thm:explicit} (see Section~\ref{sec:explicit-const} for more details).

It is possible that with sufficient time and computing resources (i.e., by making $m$ large and $\varepsilon$ tiny), we can use Theorem~\ref{thm:informalmain} to obtain embeddings with rate approaching $\alpha$, where $\alpha$ is the relative distance parameter of the misaligner. From the current computer search, we speculate that $\alpha$ is above 0.2. Thus, it is possible in the future that we have an isometric embedding where $n$ bits are mapped only to $5n$ bits.

As remarked earlier, locally self-matching strings are closely related to synchronization strings \cite{haeupler2017synchronization}. Both these objects are constructed using the algorithmic Lovasz Local Lemma \cite{moser2010constructive}. Prior work established that for all positive integers $n$ there exists a $\varepsilon$-locally self-matching string $w$ of length $n$ 
 over alphabet of size $O(1/\varepsilon^2)$ \cite{cheng2018synchronization}.    However, we could not use any of the existing analyses using the Lovasz Local Lemma in the synchronization strings literature due to the large hidden constants in those works. Thus, en route to proving Theorem~\ref{thm:explicit}, we also improve the analysis of constructing synchronization strings. To the best of our knowledge, prior to our work, the hidden constant in the alphabet size bound in these analyses was about $4900\cdot e^2\approx 3.6\times 10^4$, and we have reduced it to a quantity that approaches $e^2\approx 7.39$ as $\varepsilon$ goes to 0.

\begin{theorem}
Let $\Sigma$ be a finite set and $\varepsilon\in \left(0,\frac{1}{2}\right]$ such that the following holds: $$|\Sigma|\ge \frac{e^2}{\varepsilon^2}\cdot \left(1+4\sqrt[4]{\varepsilon}\right).$$
Then, for all positive integers $n$, there exists a $\varepsilon$-locally self-matching string $w$ over $\Sigma$ of length $n$.     
\end{theorem}

As an immediate consequence of this improved analysis, we also obtain the following (proved in Appendix~\ref{sec:sync-4}).

\begin{corollary}
\label{cor:sync-4}
There exists an infinite $0.999606$-synchronization string over an alphabet of size four.
\end{corollary}

The above result adds to the work of \cite{cheng2018synchronization} where it was shown that there is some unspecified constant $\varepsilon$ such that $\varepsilon$-synchronization strings exist over an alphabet of size four.

\subsection{Structure of Isometric Embeddings and Impossibility Results}\label{sec:intro-padding}

The embedding in Theorem~\ref{thm:informalmain} is an example of an \textit{interleaved embedding}, where the output string is formed by interleaving the input bits with sequences of bit strings that do not depend on the input. A natural question is whether there exist isometric embeddings, potentially with better rates, that do not follow this interleaving framework.  We answer this negatively, demonstrating that isometry necessitates an interleaved structure in the output strings.

\begin{theorem}[Isometry implies Interleaving]
\label{thm:informal-isometry-implies-interleaving}
    Every isometric embedding \(\varphi:\{0, 1\}^n\to \{0, 1\}^N\) of the Hamming metric into the edit metric must be an interleaved embedding.
\end{theorem}

A precise formulation of Theorem~\ref{thm:informal-isometry-implies-interleaving} requires a formal definition of interleaved embeddings, which we provide in Section~\ref{sec:isometry-implies-interleaving}. Informally, we naturally extend the intuitive notion of interleaving input bits with fixed bit patterns by allowing two additional flexibilities. First, we allow the input bits to appear \textit{out-of-order} --- e.g., the second input bit might appear after the first input bit in the output string. We further allow some of the input bits to appear \textit{complemented} in the output string. Thus, an isometric embedding of the Hamming metric into the edit metric is completely determined by three components --- the order in which input bits appear in the output, the subset of input bits that are flipped, and the sequence of fixed bit strings that are to be interleaved with the input bits.

This structural constraint on Hamming-to-edit isometric embeddings enables us to derive bounds on the achievable rate. In particular, we show that for binary strings, any isometric embedding must stretch the input strings by a factor greater than 2.133.

\begin{theorem}
\label{thm:binary-rate-upper-bound}
    There exists a positive integer \(n_0\) such that every isometric embedding \(\varphi:\{0, 1\}^n\to \{0, 1\}^N\) of the Hamming metric into the edit metric with \(n\geq n_0\) must have rate at most \(\frac{15}{32}\).
\end{theorem}

We note that the above upper bound on rate has been improved by Bhattacharya (see Theorem 4.30 in \cite{bhattacharya2025string}) to  $\frac{3}{7}+o(1)$.

Next, we generalize Theorem~\ref{thm:binary-rate-upper-bound} to larger alphabets by first extending the definition of interleaved embeddings to larger alphabets (see Section~\ref{sec:isometry-implies-generalized-interleaved-embedding}), then proving that every isometric embedding must be an interleaved embedding,  and finally showing an analogous statement to Theorem~\ref{thm:binary-rate-upper-bound}.

\begin{theorem}
\label{thm:general-rate-upper-bound}
    For every alphabet $\Sigma$, there exists an integer $n_0$ such that every isometric embedding $\varphi: \Sigma^n \to \Sigma^N$ of the Hamming metric into the edit metric with $n\geq n_0$ must have rate at most $\frac{1}{2}-\frac{1}{16|\Sigma|}$.
\end{theorem}

\subsection{Isometric Embeddings over Larger Alphabets and Surpassing the \(\nicefrac{1}{2}\) Rate Barrier}

Having established the rate limitations for isometric embeddings over the binary alphabet, we next explore if leveraging a larger alphabet $\Sigma$ enables the construction of embeddings $\varphi:\Sigma^n\to\Sigma^N$
  with potentially improved rates. As in the binary case, we can obtain high-rate embeddings over \(\Sigma\) by using misaligners for larger alphabets. In fact, it is not hard to see that misaligners designed for the binary alphabet also function as misaligners for larger alphabets, without any loss in parameters. Therefore, increasing the alphabet size can only improve the achievable rate.

We show that even without relying on the full machinery of misaligners, one can obtain isometric embeddings with rates arbitrarily close to $\frac{1}{3}$ by making the alphabet large enough.

\begin{theorem}
\label{thm:approaching-1/3-is-possible}
    For any $\rho>0$, there exists an alphabet $\Sigma$ such that for every positive integer \(n\), there is an isometric embedding $\varphi_n:\Sigma^n\to\Sigma^N$ of the Hamming metric into the edit metric with rate at least \(\frac{1}{3}-\rho\). 
\end{theorem}

Theorem~\ref{thm:general-rate-upper-bound} tells us that as long as the input and output alphabets remain the same, the rate in Theorem~\ref{thm:approaching-1/3-is-possible} can never be improved to \(\frac{1}{2}\) or anything better. But what if we allow the input and output alphabets to be \textit{different}? 

In this setting, our notion of rate needs to be refined. In particular, when using the larger output alphabet, the appropriate measure to consider is not the lengths of strings but the \textit{number of bits} contained in them. Thus, for an embedding $\varphi:\Sigma_{\textnormal{in}}\to \Sigma_{\textnormal{out}}$, we define the rate as $\frac{n\log|\sigmain|}{N\log|\sigmaout|}$. With this refined definition, we show that not only can the \(\frac{1}{2}\) rate barrier be surpassed, but the rate can actually be made \textit{arbitrarily close to 1}!

\begin{theorem}
\label{thm:approaching-1-is-possible}
        For every $\rho>0$, there exist alphabets $\sigmain$, $\sigmaout$ such that for all positive integers \(n\), there is an isometric embedding $\varphi_n:\sigmain^n\to \sigmaout^{N}$ of the Hamming metric into the edit metric with rate at least \(1-\rho\).
\end{theorem}

We provide a proof of Theorem~\ref{thm:approaching-1-is-possible} in Section~\ref{sec:close-to-1} while Theorem~\ref{thm:approaching-1/3-is-possible} is proven in Section~\ref{sec:close-to-1/3}.

\subsection{Related works}
In this subsection, we survey three lines of related works.

\paragraph{Embedding Hamming Metric into Edit Metric.} As mentioned earlier, embeddings of the Hamming metric into the edit metric have primarily been studied to establish hardness results for problems under the edit distance. The first known instance of this, to the best of our knowledge, is due to Rubinstein~\cite{Rubinstein18}, who gave a near-isometric embedding for binary strings with rate \(\Theta(1/\log n)\) and used it to prove the hardness of the approximate closest pair problem under edit distance. In a subsequent work, Cohen-Addad \textit{et al.}~\cite{cohen2019lower} presented a fully isometric embedding (in contrast to Rubinstein's near-isometry), albeit at the cost of a larger output alphabet. Their construction achieved a rate of \(\Theta(1/(\log n \log \log n))\). The folklore isometric embedding referenced in Section~\ref{sec:quest-for-constant} explicitly appears in~\cite{ABCSS23}.

\paragraph{Embedding Edit metric into Hamming metric.} Embeddings in the other direction, i.e., from the edit metric to the Hamming metric, are abundant in the literature. A key motivation for embedding the edit metric into the Hamming metric is computational: Hamming distance can be computed efficiently in linear time, whereas edit distance cannot be solved in sub-quadratic time unless the Strong Exponential Time Hypothesis is false~\cite{BaIn18}.  This underscores the need for approximation algorithms that can run in subquadratic time. One potential approach is to develop an embedding with low distortion, where the embedding process itself can be performed in subquadratic time.

The seminal work by Ostrovsky and Rabani~\cite{OsRa07} shows an embedding of the edit metric into the Hamming metric with distortion $2^{O(\sqrt{\log n \log \log n})}$. The rate achieved is  $1/\log n$, and the embedding can be computed in polynomial time. There has been extensive research aimed at establishing lower bounds for the distortion of any such embedding, with the best-known lower bound being $O(\log n)$~\cite{KrRa09}.

An alternative approach to potentially improve the results above is through randomized embeddings. In this approach, the embedding function takes both an input string and a random string, and we measure the distortion when two strings are mapped using the same random string sequence. Within this framework, Johwari~\cite{Jowhari12} showed the existence of $O\left(\log n \log^* n\right)$-distortion. Subsequently, ~\cite{CGK16}  were the first to eliminate the dependence of the dimension in the distortion, achieving a quadratic distortion, i.e., if the edit distance between the input strings is at most $k$, then the Hamming distance between the embedded strings is bounded by $O(k^2)$.

\paragraph{Synchronization strings.} \begin{sloppypar}Synchronization strings were introduced by Haeupler and Shahrasbi~\cite{haeupler2017synchronization} in the context of codes able to tolerate insertion and deletion errors. In our context, the original application of synchronization strings can be viewed as embedding the Hamming metric near isometrically into the edit metric by increasing the alphabet size by a constant factor. Throughout the years, these strings have found uses in many settings, including interactive coding~\cite{haeupler2017synchronization2}, locally decodable insertion-deletion codes~\cite{haeupler2017synchronization3}, and list decodable insertion-deletion codes~\cite{haeupler2018synchronization4}. For a survey of the many uses of synchronization strings and how they are constructed, the reader is referred to~\cite{sync-survey}.  \end{sloppypar}

\subsection{Organization of Paper}
In Section~\ref{sec:overview}, we provide an overview of the proofs for our main results. In Section~\ref{sec:prelim}, we introduce some notations and definitions of relevance to this paper. In Section~\ref{sec:misaligner}, we formally define misaligners -- the key ingredient to our high-rate isometric embeddings. In Section~\ref{sec:locally-self-matching}, we introduce locally self-matching strings and give improved bounds on alphabet size for such strings. In Sections~\ref{sec:embedding} and~\ref{sec:analysis}, we show how to combine misaligners and locally self-matching strings to obtain high-rate isometric embeddings for binary strings. In Section~\ref{sec:explicit-const}, we describe how to find a misaligner using computer search and give an explicit rate \(\frac{1}{8}\) isometric embedding for binary strings, proving Theorem~\ref{thm:explicit}. In Section~\ref{sec:close-to-1/3}, we consider embeddings over larger alphabets and prove Theorem~\ref{thm:approaching-1/3-is-possible}. Then Section~\ref{sec:isometry-implies-interleaving} is dedicated to defining interleaved embeddings and proving Theorem~\ref{thm:informal-isometry-implies-interleaving}. In Section~\ref{sec:UpperBoundsOnTheRate}, we derive rate upper bounds for isometric embeddings, proving Theorem~\ref{thm:general-rate-upper-bound}. Finally, in Section~\ref{sec:close-to-1}, we show how to surpass the \(\frac{1}{2}\) rate barrier by considering different-sized input and output alphabets and prove Theorem~\ref{thm:approaching-1-is-possible}.

\section{Proof Overview}\label{sec:overview}

In Section~\ref{sec:sync-proof-overview}, we provide a proof overview of Theorem~\ref{thm:const}, which is our main conceptual contribution. We do not provide a formal proof of Theorem~\ref{thm:const} as it is subsumed by Theorem~\ref{thm:informalmain} and Theorem~\ref{thm:explicit}. In Section~\ref{sec:overview-2}, we provide a proof overview of  Theorem~\ref{thm:informalmain}, which is one of our main technical contributions. Furthermore, in Section~\ref{sec:overview-3}, we summarize the key ideas behind our proof of Theorem~\ref{thm:informal-isometry-implies-interleaving}, while in Section~\ref{sec:overview-4}, we provide a proof sketch of Theorem~\ref{thm:general-rate-upper-bound}. 

\subsection{Constant Rate Embeddings via Synchronization Strings}
\label{sec:sync-proof-overview}

Our first observation is that if we are allowed an unbounded output alphabet, then there is a straightforward isometric embedding that takes any input string and interleaves it with a string with all distinct symbols. This embedding is clearly isometric since any misalignment of these new symbols results in mismatches, increasing the edit cost beyond that of the simple identity alignment, which equals the Hamming distance between the input strings. Moreover, the length of the output strings produced by this embedding is only twice that of the input strings.

However, our goal is to have a binary output alphabet, not an unbounded one. So, as an intermediate step, let us limit ourselves to a large but constant-sized output alphabet. Now, if we want to mimic our earlier embedding, we would need to find a string over a \textit{constant-sized} alphabet that \textit{approximates} the string with all distinct symbols. It turns out that objects known as \textit{synchronization strings}~\cite{haeupler2017synchronization} have the exact property that we want. Synchronization strings have many equivalent definitions. The one that is useful for our discussion is the following --- a string is said to be an $\varepsilon$-synchronization string if, for every substring, its relative self-edit distance\footnote{The self-edit distance of a string is the cost of the cheapest alignment converting the string to itself without matching any character to itself.} is at least $1-\varepsilon$.
For every value of $\varepsilon$, one can efficiently construct $\varepsilon$-synchronization strings  over an alphabet of size  $\Theta(1/\varepsilon^2)$~\cite{cheng2018synchronization}.  We now illustrate how these strings can be used to construct constant-rate isometric embeddings over a large but fixed-size output alphabet.

To embed a string of length $n$, consider a $\nicefrac{2}{3}$-synchronization string of length $3n$ over an alphabet $\Sigma$ that is disjoint from the input alphabet $\{0, 1\}$. The embedding process, applied to an input string of length $n$, involves inserting $3$ symbols from the synchronization string after every input symbol. It is evident that both the rate and the output alphabet size remain constant.

We now argue that this gives an isometric embedding. Suppose it does not; then there must be a pair of strings such that the edit distance between their embeddings is strictly smaller than their Hamming distance. Consider an optimal edit alignment of the embedded strings. It can be shown that this alignment partitions the strings into intervals, where for each interval in the partition, either the entire interval is mapped using the identity mapping, i.e., ``vertical'', or it is self-aligned such that none of its characters are mapped to themselves, i.e., ``nowhere-vertical''. See Figure \ref{fig:alignment} for an illustration.
\begin{figure}[ht]
    \centering
    \includegraphics[width = 0.7\textwidth]{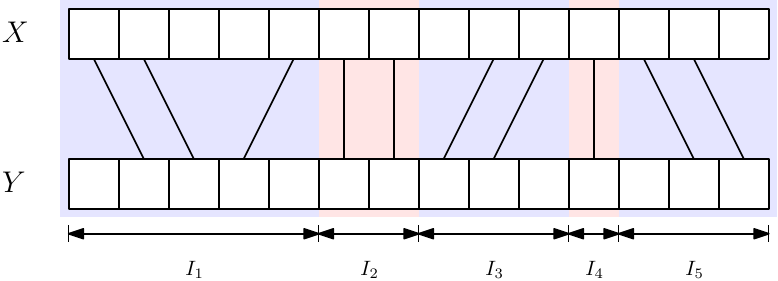}
    \caption{An optimal alignment converting the embedded string $X$ to the embedded string $Y$. Note the alternating maximal nowhere-vertical and vertical intervals (highlighted blue and red, respectively) $I_1, I_2, I_3, I_4$ and $I_5$.}
    \label{fig:alignment}
\end{figure}

Now, if the edit distance is indeed strictly less than the Hamming distance, then there is at least one interval, say of length $\ell$, in this partition where none of its characters are mapped to themselves, but the edit cost in that interval is strictly less than the Hamming cost, which is at most $\nicefrac{\ell}{4}$. We then show that this implies the existence of a substring in the synchronization string with a relative self-edit distance less than $\nicefrac{1}{3}$,  which is a contradiction.  This substring can be found by considering the aforementioned interval and deleting all the characters that come from the input strings. This results in a substring with length $\ell' := \nicefrac{3\ell}{4}$. One can also find a self-alignment of this substring by taking the original alignment, keeping it unchanged except for positions where some character in the substring was matched to some character in one of the input strings. In those positions, we simply apply a deletion. It is not hard to see that the cost of this alignment remains unchanged and is thus less than $\nicefrac{\ell}{4} = \nicefrac{\ell'}{3}$, which contradicts the fact that the string used to pad is a $\nicefrac{2}{3}$-synchronization string. In Figure~\ref{fig:selfLCS}, the red characters are part of some $\varepsilon$-synchronization string, and the green ones are part of some input string.

\begin{figure}[ht]
    \centering
    \includegraphics[width = 0.6\textwidth]{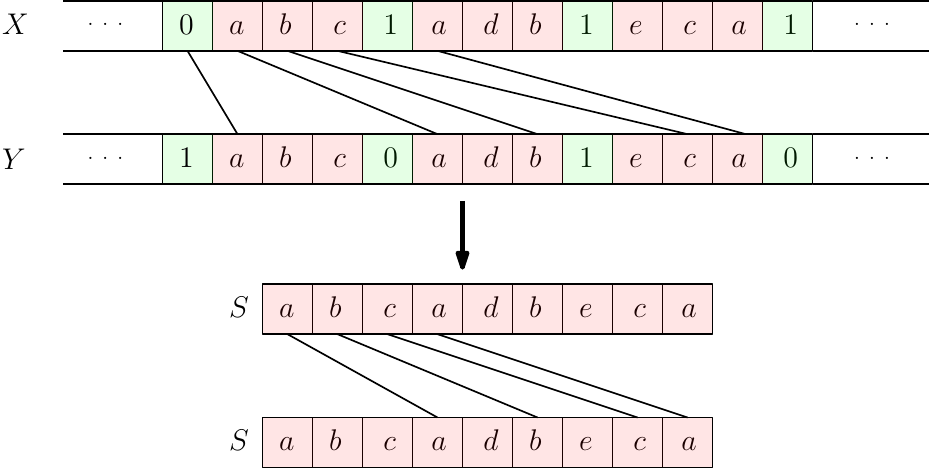}
    \caption{A nowhere-vertical alignment on some interval in the embedded strings $X$ and $Y$ implies a self-alignment of a substring $S$ of the $\varepsilon$-synchronization string with the same cost.}
    \label{fig:selfLCS}
\end{figure}

One can then bring the output alphabet size down to two using the following alphabet reduction procedure --- we replace each character of the synchronization string with its binary encoding (of length $\log |\Sigma|$) along with an opening and closing block of $0$s and $1$s with lengths equal to that of the binary encoding.  The final resulting rate is still constant. However, the large alphabet sizes required to construct $\varepsilon$-synchronization strings mean that the resulting rate is far from optimal. With the specific setting of $\varepsilon = \nicefrac{2}{3}$ and known bounds for alphabet sizes for synchronization strings, one gets a rate of about $\nicefrac{1}{153}$ from this.

\subsection{Optimizing the Rate: Enter Misaligners!}
\label{sec:overview-2}
To improve the rate of our embedding, we move away from the framework of interleaving a synchronization string over a larger alphabet and then performing an alphabet reduction. Instead, we utilize a tailored construction that gives us a much larger rate.

The key ingredient in our construction is an object called a \textit{misaligner}, which the reader should think of as a code in the edit metric with robust distance guarantees. A misaligner consists of a set of fixed-length codewords over the alphabet $\{0, 1, \star\}$, where $\star$ is a special wildcard symbol. Each codeword contains a wildcard symbol at every $t^{\text{th}}$ position, where $t$ is a parameter that controls the rate of the embedding. Thus, if one concatenates a sequence of codewords from the misaligner, one obtains a long binary string with wildcard symbols at every $t^{\text{th}}$ position. To embed some input string, one simply substitutes these wildcard symbols with symbols from the input string. The rate of this embedding is $\nicefrac{1}{t}$ by definition.

To fully describe the embedding, one also needs to specify the \textit{order} in which codewords from the misaligner are concatenated. To do this, we make use of a variant of synchronization strings, which we call \textit{locally self-matching strings}. Given $\varepsilon \in (0, 1)$, an $\varepsilon$-locally self-matching string is one where every substring of length \(\ell\) has at most $\varepsilon\ell$ non-crossing symbol matches to itself if one is not allowed to match symbols vertically.\footnote{Alternatively, every substring has ``nowhere-vertical'' relative LCS at most \(\varepsilon\); see Section~\ref{sec:locally-self-matching} for a formal definition.} Given both these ingredients, we can now describe our embedding. To embed an input string, we first start with a sufficiently long locally self-matching string over some large alphabet. Next, we replace each symbol of this string with a codeword from the misaligner. Note that in order to do this, the number of codewords in the misaligner must be at least as large as the alphabet of the locally self-matching string. Finally, we substitute the wildcard symbols in this string with symbols of the input string to obtain the embedded string. Our goal is to show that by choosing $t$ sufficiently large, the properties of the locally self-matching string and the misaligner guarantee that the resulting construction is an isometric embedding.

For concreteness, in the remainder of this section, let us assume that a 0.1-locally self-matching string $s$ exists over some large alphabet $\Sigma$. We also assume that a misaligner $\mathcal{C}$ with $|\Sigma|$ many codewords exists. The other properties of this misaligner will be specified later. We will show how the properties of $s$ and $\mathcal{C}$ will guarantee isometry.

Let us call a codeword with its wildcard symbols instantiated a \textit{block}. Assume for the sake of contradiction that the embedding we just described is not an isometric embedding. As before, this implies the existence of an interval with an edit cost that is strictly lower than its Hamming cost. The first guarantee that the misaligner provides is that this interval must contain at least two full blocks. This is ensured by enforcing the codewords of the misaligner to have the following property --- for any substring arising from the concatenation of multiple codewords and not containing two full blocks must have the edit distance equal to the Hamming distance, no matter how the wildcard symbols are instantiated. See Figure~\ref{fig:shortIntervals} for an illustration of this property.
\begin{figure}[ht]
    \centering
    \includegraphics[width = 0.6\textwidth]{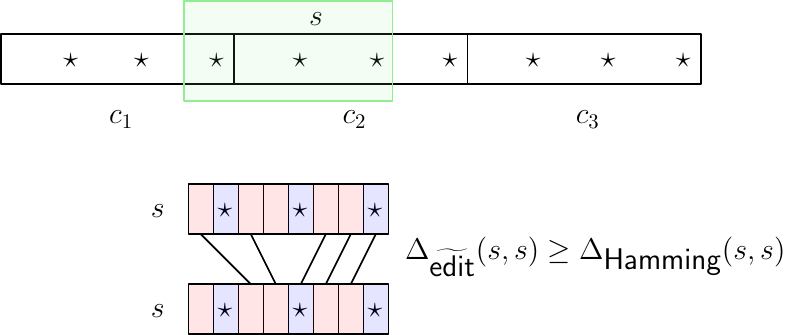}
    \caption{For short intervals, the misaligner guarantees isometry: any nowhere-vertical edit alignment must pay at least as much as the Hamming distance, no matter how the wildcards are instantiated.}
    \label{fig:shortIntervals}
\end{figure}

For larger intervals (i.e., intervals containing at least two full blocks), we assume for simplicity that the interval consists solely of a sequence of full blocks with no partial blocks at the boundaries. The analysis then goes on to assess the cost of the optimal alignment by decomposing it into contributions from individual blocks. Our objective is to demonstrate that, for most blocks, there is a substantial cost, which, for the purposes of this discussion, is 0.2 times the length of a block. 

Our first observation concerns any pair of blocks within the specified interval that are instantiations of the same codeword but appear at different positions in the embedded strings. It can be shown (Obervation~\ref{obs:bad-blocks}) that we can assume that in our optimal alignment, exactly one of the following cases occurs: either all of the characters of one of the blocks are matched ``vertically'' to all of the characters in the other block\footnote{By ``vertically'', we mean the $i^{\text{th}}$ character of the first block is matched against the $i^{\text{th}}$ character of the second.} (in which case, the first block is referred to as ``bad''; see Figure~\ref{fig:badBlocks}), or none of the characters are matched ``vertically''. From the set of bad blocks and their matches, one can recover a sequence of nowhere-vertical matches in the original locally self-matching string $s$. Since $s$ was a 0.1-locally self-matching string, one can thus conclude that the fraction of bad blocks is at most 0.1.

\begin{figure}[ht]
    \centering
    \includegraphics[width = 0.6\textwidth]{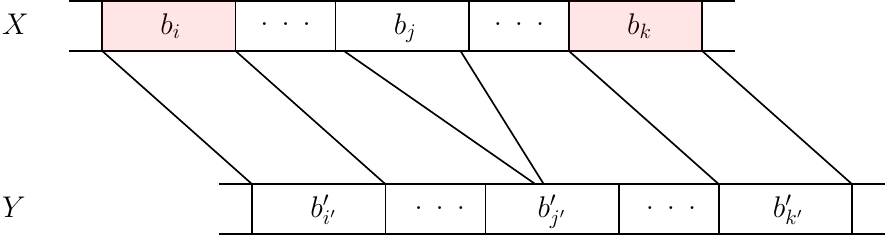}
    \caption{The blocks $b_i$ and $b'_{i'}$ as well as $b_j$ and $b'_{j'}$ and $b_k$ and $b'_{k'}$ are instantiations of the same codeword and appear in different positions in $X$ and $Y$. The bad blocks --- those where all characters are matched vertically ---  are highlighted in red.}
    \label{fig:badBlocks}
\end{figure}


Next, let us analyze the remaining blocks. The optimal alignment transforms each of these blocks into a substring coming from the concatenation of multiple blocks. If the size of this substring exceeds $0.2$ times the block size, then the cost will also exceed $0.2$. For the remaining blocks, the misaligner must ensure that the edit cost of each block remains substantial. This can be accomplished by satisfying the following requirements:

Consider a block $b$ and take \textit {any} substring $s$ of roughly the same size, composed of partial blocks.  If none of these blocks is equal to $b$, we require that the relative edit cost between $b$ and $s$ is at least $0.2$. If one of these blocks is equal to $b$, we require that the edit cost of alignments with no vertical edges between $b$ and its counterpart in $s$ is at least $0.2$. See Figure~\ref{fig:blockVSsubstring} for an illustration.

\begin{figure}[ht]
    \centering
    \includegraphics[width = 0.7\textwidth]{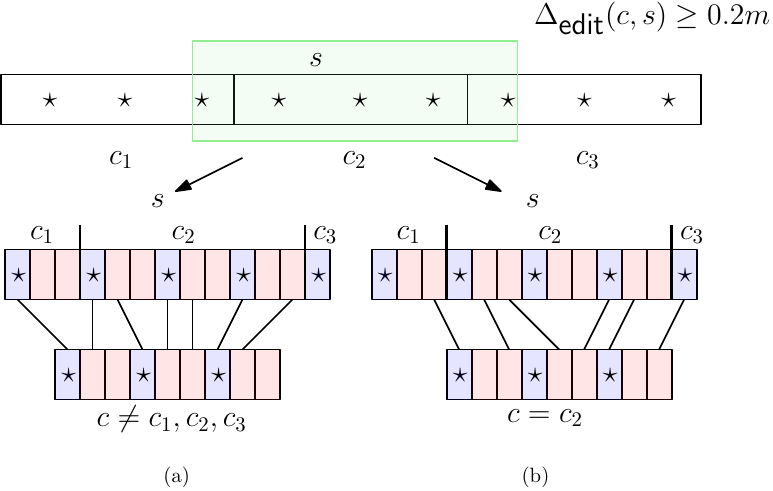}
    \caption{The alignment transforms every block $c$ into a substring $s$ arising from the concatenation of multiple blocks. If $c$ and $s$ are roughly the same size, the misaligner guarantees the relative edit distance between $c$ and $s$ is at least 0.2. Observe that in case (b), the  alignment has no vertical edges between $b$
    and its counterpart in $s$.}
    \label{fig:blockVSsubstring}
\end{figure}

Since bad blocks cost zero, the total cost is given by the sum of costs incurred by the non-bad blocks. Thus, by the properties of the misaligner and the locally self-matching string, the relative cost is lower bounded by $(1-0.1)\times0.2=0.18$. Therefore, if we choose $t=6$, we arrive at a contradiction, as the cost would then exceed the number of wildcards, which is an upper bound on the Hamming distance in the interval between any two strings. Consequently, we can achieve a  $\nicefrac{1}{6}$-rate isometric embedding. The case where the interval may contain partial blocks introduces additional requirements, which are detailed in Section~\ref{sec:misaligner}.

\subsection{Unveiling the Interleaved Structure of Isometric Embeddings}
\label{sec:overview-3}

In this section, we outline the key ideas behind our proof of Theorem~\ref{thm:informal-isometry-implies-interleaving}, which states that every isometric embedding of the Hamming metric into the edit metric is necessarily an interleaved embedding. The formal definition of interleaved embeddings is deferred to Section~\ref{sec:isometry-implies-interleaving}. Here, we instead fix an arbitrary isometric embedding \(\varphi:\{0, 1\}^n\to \{0, 1\}^N\) and make a series of observations that reveal certain constraints on \(\varphi\), ultimately uncovering its interleaved nature.

We start by fixing an input string \(x\in \{0, 1\}^n\) and ask what happens to \(\varphi(x)\) as we flip a single bit in \(x\). Since \(\varphi\) is isometric, flipping a single bit in \(x\) must correspond to exactly one bit flip in \(\varphi(x)\). This allows us to define a function \(\eta_x:[n]\to[N]\), \(\eta_x(i)\) is the unique index in \(\varphi(x)\) that is flipped if the \(i^{\text{th}}\) bit of \(x\) is flipped. 

Next, we show that the function \(\eta_x\) is \textit{injective}, meaning that different bit flips in \(x\) affect different positions in the output \(\varphi(x)\). This is not hard to see --- if two different bit-flips in \(x\) affected the same position in \(\varphi(x)\), then we would obtain a pair of strings at Hamming distance 2 whose images under \(\varphi\) have edit distance 0, contradicting the fact that \(\varphi\) is isometric.

A crucial step in the proof is showing that \(\eta_x(i)\) does not depend on \(x\). That is, once \(i\) is fixed, flipping the \(i^{\text{th}}\) bit in any input string always affects the same location in the output. This means the function \(\eta_x\) is the same for all inputs and can simply be written as \(\eta\). This can be proven by fixing a pair of input strings \(x, y\) and an index \(i\) and inducting on the Hamming distance between \(x\) and \(y\). This key step is presented in detail in Section~\ref{sec:isometry-implies-interleaving}.

Finally, we show that \(\varphi(x)[\eta(i)]\) depends only on \(x[i]\), meaning that each input bit is mapped consistently to the same output position, either preserving or flipping its value. With this, the interleaved nature of \(\varphi\) is revealed: the output positions fall into two categories --- those that directly correspond to input bits (determined by \(\eta\)) and the remaining positions that are independent of the input. For the full details, we refer to Section~\ref{sec:isometry-implies-interleaving}.

\subsection{Upper Bounds on the Rate}

\label{sec:overview-4}

In this section, we outline the proof of Theorem~\ref{thm:general-rate-upper-bound}, which establishes that any isometric embedding of the Hamming metric into the edit metric must have a rate strictly bounded away from \(\frac{1}{2}\). More precisely, we aim to show that for any alphabet \(\Sigma\), there exists some \(\varepsilon>0\) such that for sufficiently large \(n\), no isometric embedding \(\varphi:\Sigma^n\to \Sigma^N\) can achieve a rate exceeding \(\frac{1}{2}-\varepsilon\).

Our approach is by contradiction: we assume that for all \(n\), there exists an isometric embedding \(\varphi\) with a rate exceeding the claimed bound. The goal then is to show that as \(n\) grows, \(\varphi\) cannot remain isometric. Specifically, we seek a pair of strings such that the edit distance between their embeddings is strictly smaller than their Hamming distance, contradicting isometry. Rather than constructing such a pair explicitly, we will argue its existence via the probabilistic method.

The key idea is to consider a collection \(\mathcal{C}\) of carefully chosen triples \((x, y, \mathcal{A})\), where \(x, y\in \Sigma^n\) and \(\mathcal{A}\) is an edit distance alignment between \(\varphi(x)\) and \(\varphi(y)\). We then select a random triple from \(\mathcal{C}\) and show that the \emph{expected} cost of \(\mathcal{A}\) with respect to the embeddings \(\varphi(x)\) and \(\varphi(y)\) is too low, ensuring the existence of a triple in \(\mathcal{C}\) that violates the isometry of \(\varphi\).

\paragraph{Description of \(\mathcal{C}\).} The alignments in \(\mathcal{C}\) are simple ``shift'' alignments, where each symbol in one string is matched with a symbol in the other string shifted by a fixed amount. More formally, for a non-zero integer \(\delta\), we define \(\mathcal{A}_\delta\) as the alignment where each symbol in the first string is paired with a symbol in the second string shifted by \(\delta\) positions. We consider alignments of the form \(\mathcal{A}_\delta\) where \(|\delta|\) is bounded by a constant that depends on the rate we seek to rule out, in particular, \(\varepsilon\), but is independent of \(n\).

Since \(\varphi\) is isometric, by the generalization of Theorem~\ref{thm:informal-isometry-implies-interleaving} to larger alphabets (i.e., Theorem~\ref{thm:isometry-implies-generalized-interleaving}), it must be an interleaved embedding. This implies that some symbols in \(\varphi(x)\) are ``frozen'' (independent of \(x\)), while the remaining symbols are ``mutable'' (determined by \(x\)). Using this observation, for each alignment \(\mathcal{A}_\delta\), we construct a pair of strings \(x_\delta, y_\delta\in \Sigma^n\) satisfying the following two properties.

\begin{itemize}
    \item The Hamming distance between \(x_\delta\) and \(y_\delta\) is \(n\).
    \item No mutable symbol in \(\varphi(x_\delta)\) is substituted under \(\mathcal{A}_\delta\); they are either matched or deleted.
\end{itemize}

Section~\ref{sec:UpperBoundsOnTheRate} details how to construct such pairs given \(\mathcal{A}_\delta\). The final collection \(\mathcal{C}\) consists of all triples \((x_\delta, y_\delta, \mathcal{A}_\delta)\) for \(\delta\) in a constant-sized set of shifts.

\paragraph{Bounding the Expected Alignment Cost.} We now pick a random \((x_\delta, y_\delta, \mathcal{A}_\delta)\) from \(\mathcal{C}\) and analyze the expected cost of \(\mathcal{A}_\delta\) with respect to \(\varphi(x_\delta)\) and \(\varphi(y_\delta)\). Since \(\mathcal{A}_\delta\) is a shift alignment, its cost due to insertions and deletions is at most a constant. Moreover, by our choice of \(x_\delta\) and \(y_\delta\), there are no substitutions involving mutable symbols --- only frozen symbols contribute to substitutions. Now if the rate of \(\varphi\) is too high, the number of frozen symbols must be small, which in turn means the expected number of substitutions is also small. We show that for sufficiently large \(n\), the expected cost of \(\mathcal{A}_\delta\) falls below \(n\), leading to a contradiction. For the detailed calculations along with choices for the maximum value of \(\delta\), the reader is referred to Section~\ref{sec:UpperBoundsOnTheRate}.

\section{Preliminaries}\label{sec:prelim}

\paragraph{Intervals and Strings.}
For any positive integer $n$, we define $[n]:= \{1, 2, \ldots , n\}$. For positive integers $i, j$, we define the interval $[i, j] := \{i, i+1, \ldots , j\}$. We define a partial order on the set of all intervals in the following natural way --- given intervals $I=[i, j]$ and $I'=[i', j']$, we say $I< I'$ if and only if $j < i'$.

A string $x = x[1]x[2]\cdots x[n]$ over alphabet $\Sigma$ is a sequence of $n$ symbols from $\Sigma$. We denote by $|x|$ the length of the string $x$. Given a string $x$ and an interval $I = [i, j]\subseteq [|x|]$, we denote by $x_{|I}$ the substring $x[i]x[i+1]\cdots x[j-1]x[j]$. We further extend this notation to arbitrary sets of indices and not just intervals. Given any subset $S=\{i_1, i_2, \ldots , i_k\}$ with $1\leq i_1 < i_2 < \cdots <i_k \leq |x|$, we denote by $x_{|S}$ the string $x[i_1]x[i_2]\cdots x[i_k]$. We denote by $x^{\mathcal{R}}$ the reverse of the string $x$.

\paragraph{Edit distance alignments.}
Let $x, y\in \{0, 1\}^*$ be two strings of length $n$ and $m$, respectively, such that $n\geq m$. An edit distance alignment between $x$ and $y$ is a set $\mathcal{A}=\{(i_1, j_1), (i_2, j_2), \ldots , (i_k, j_k)\}$ with $0\leq k \leq m$ such that $1\leq i_1 < i_2 < \cdots < i_k \leq n$ and $1\leq j_1 < j_2 < \ldots < j_k \leq m$. If $(i, j) \in \mathcal{A}$, then we say $\mathcal{A}$ aligns $x[i]$ to $y[j]$. Given an edit distance alignment $\mathcal{A}$ between $x$ and $y$, we define the following three sets $S_x^y({\mathcal{A}}), D_x^y({\mathcal{A}})$ and $I_x^y({\mathcal{A}})$ as follows.
\[S_x^y({\mathcal{A}}):= \{(i, j)\in \mathcal{A} : x[i]\neq y[j]\}\]
\[D_x^y({\mathcal{A}}):=\{i\in [n] : \nexists j \in [m] \textnormal{ such that } (i, j) \in \mathcal{A}\}\]
\[I_x^y({\mathcal{A}}) := \{j\in [m] : \nexists i \in [n] \textnormal{ such that } (i, j) \in \mathcal{A}\}\]
We call the sets $S_x^y({\mathcal{A}}), D_x^y({\mathcal{A}})$ and $I_x^y({\mathcal{A}})$ the set of substitutions, deletions, and insertions, respectively, associated with $\mathcal{A}$. Additionally, we will call \(\mathcal{A}\setminus S_x^y(\mathcal{A})\) the set of matches associated with \(\mathcal{A}\). The cost of an edit distance alignment $\mathcal{A}$ (with respect to the strings $x$ and $y$), denoted by $\cost_x^y(\mathcal{A})$, is defined as ---
\[\cost_x^y(\mathcal{A}) := |S_x^y({\mathcal{A}})| + |D_x^y({\mathcal{A}})| + |I_x^y({\mathcal{A}})|\]
If the strings $x$ and $y$ are clear from context, we will often drop the subscript and the superscript, and write $\cost(\mathcal{A})$ instead.

The edit distance between $x$ and $y$, denoted by $\ed(x, y)$ is defined as the cost of the cheapest edit distance alignment between $x$ and $y$, i.e., $\ed(x, y) = \min_{\mathcal{A}}\cost(\mathcal{A})$. An alignment $\mathcal{A}$ is called \textit{nowhere-vertical} if for every $i\in [n]$, $(i, i)\notin \mathcal{A}$. We will often use the notion of the \textit{nowhere-vertical edit distance} of $x, y$, denoted by $\selfed(x, y)$, which is defined to be the cost of the cheapest nowhere-vertical edit distance alignment between $x$ and $y$. Note that for every string $x$, $\selfed(x, x)$ is non-zero --- in fact, it is at least two. 

We call an edit distance alignment $\mathcal{A}$ between two string $x$ and $y$ a \textit{common subsequence alignment} if $S_x^y(\mathcal{A}) = \emptyset$, i.e., the alignment induces no substitutions. The length of the longest common subsequence between $x$ and $y$, denoted by $\lcs(x, y)$, is given by $\max_{\mathcal{A}} (|x|+|y|-\cost(\mathcal{A}))$ as $\mathcal{A}$ ranges over all common subsequence alignments of $x$ and $y$. We can analogously define the length of the \textit{nowhere-vertical longest common subsequence} between $x$ and $y$, denoted by $\slcs(x, y)$, to be the maximum value of $(|x|+|y|-\cost(\mathcal{A}))$ as $\mathcal{A}$ ranges over all nowhere-vertical common subsequence alignments of $x$ and $y$.

\paragraph{Binary strings with wildcard symbols.} A \textit{binary string with wildcards} is a string over the alphabet $\{0, 1, \star\}$. Let $w$ be a binary string with wildcards. We denote by $\Gamma(w)$ the number of wildcard symbols $\star$ in $w$. For a binary string $x\in \{0, 1\}^{\Gamma(w)}$, the \textit{instantiation of $w$ by $x$}, denoted by $w_x$, is the string obtained by replacing, for each $i\in [\Gamma(w)]$, the $i^{\textnormal{th}}$ occurrence of $\star$ in $w$ by the $i^\textnormal{th}$ symbol of $x$. The edit distance between two binary strings $w$ and $w'$ with wildcards, denoted by $\ed(w, w')$ by overloading the notation for edit distance between strings without wildcards, is the minimum edit distance between any two instantiations of them. More precisely, we define ---
\[\ed(w, w') := \min_{\substack{{x \in \{0, 1\}^{\Gamma(w)}}\\{x' \in \{0, 1\}^{\Gamma(w')}}}} \ed\left(w_x, w'_{x'}\right)\]
Similarly, the nowhere-vertical edit distance between two binary strings $w$ and $w'$ with wildcards, $\selfed(w, w')$, is defined as ---
\[\selfed(w, w') := \min_{\substack{{x \in \{0, 1\}^{\Gamma(w)}}\\{x' \in \{0, 1\}^{\Gamma(w')}}}} \selfed\left(w_x, w'_{x'}\right)\]

\section{Misaligners}\label{sec:misaligner}

In this section, we formally define misaligners, which can be thought of as codes in the edit metric with robust distance guarantees. A misaligner is robust in two distinct senses. First, all codewords in a misaligner contain wildcard symbols, and no matter how these wildcard symbols are instantiated, the misaligner ensures large distance between codewords. The second guarantee is even stronger: not only is the distance between pairs of codewords large; even if one compares a codeword against a concatenation of multiple codewords, one observes large distance. This protects against scenarios such as the concatenation the suffix and prefix of two codewords being close to some other codeword.  

\begin{defi}[Misaligners]\label{def:misaligners}
    Given positive integers $m, k, t$ such that $t$ divides $m$ and $\alpha \in (0, \frac{1}{2}]$, an $(m, k, t, \alpha)$-misaligner is a set $\mathcal{C} \subseteq \{0, 1, \star\}^m$ with $|\mathcal{C}| = k$  such that the following properties hold.
    \begin{enumerate}
        
        \item \label{prop:rate-of-wildcards}\textbf{Rate of Wildcards:} For every $c\in \mathcal{C}$, the $i^{\textnormal{th}}$ symbol of $c$ is $\star$ if and only if $i\equiv 1 \pmod t$, i.e., every $t^{\textnormal{th}}$ symbol (starting with the first one) in $c$ is $\star$.

        \item \label{prop:short-intervals} \textbf{Short Intervals:} For every distinct    $c_{1}, c_{2}, c_{3}\in \mathcal{C}$, and every integer $\ell \leq 3m-2$, if $s$ is a substring of $c_{1}\circ c_{2}\circ c_{3}$ of length $\ell$, $\ed(s_x, s_y) = \ham(s_x, s_y)$ for all $x, y\in \{0, 1\}^{\Gamma(s)}$.
         
        \item \label{prop: full-block-distance} \textbf{Block vs. Substring:} For every distinct $c_{1}, c_{2}, c_{3}\in \mathcal{C}$, if $s$ is a substring of $c_{1}\circ c_{2}\circ c_{3}$, then the following must hold.
        \begin{enumerate}
            \item \label{prop: full-block-distance-1} For every $c\in \mathcal{C}$ that is equal to none of $c_{1}, c_{2}, c_{3}$, $\ed(c, s)\geq \alpha m$.
            \item \label{prop: full-block-distance-2} For every $c\in \mathcal{C}$ that is equal to exactly one of $c_{1}, c_{2}, c_{3}$, then any edit distance alignment between $c$ and $s$ where $c$ and its copy $c'\in \{c_1, c_2, c_3\}$ is nowhere-vertically aligned\footnote{i.e., the alignment is not allowed to align $c[i]$ with $c'[i]$ for any $i\in [m]$.} has cost at least $\alpha m$.
        \end{enumerate}

        \item \label{prop: boundary-blocks} \textbf{Block and a Half vs. Substring:} For every distinct $c_{1}, c_{2}, c_{3} \in \mathcal{C}$, the following must hold.
        \begin{enumerate}
            \item \label{prop: boundary-blocks-1} Let $s$ be a string  formed by concatenating a proper suffix  of $c_{1}$ with $c_{2}$. 
            Let $s'$
            be a string obtained by concatenating $s$ with a prefix of $c_3$.  Then $\selfed(s, s') \geq \alpha\cdot(|s|)$, i.e., any alignment where $s$ and its copy in $s'$ is nowhere-vertically aligned has cost at least $\alpha|s|$. 
            \item \label{prop: boundary-blocks-2}
            Let $s$ be a string formed by concatenating $c_{2}$ with a proper prefix of $c_{3}$. 
            Let $s'$
            be a string obtained by concatenating a suffix of $c_1$ with $s$.  Then $\selfed(s^{\mathcal{R}}, s'^{\mathcal{R}})\geq \alpha\cdot(|s|)$, i.e., any alignment where $s$ and its copy in $s'$ is nowhere-vertically aligned has cost at least $\alpha|s|$.
        \end{enumerate}
    \end{enumerate}
\end{defi}

See Figure~\ref{fig:properties} for an illustration of all the properties of a misaligner.
\begin{figure}
     \centering
     \begin{subfigure}[b]{0.35\textwidth}
         \centering
         \includegraphics[width=\textwidth]{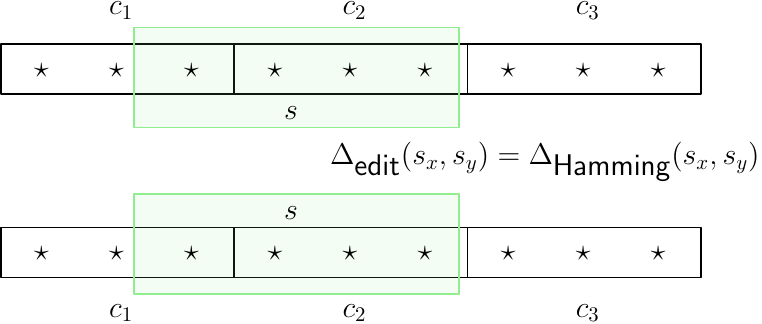}
         \caption{\textbf{Short Intervals:} For intervals of length at most $3m-2$, the misaligner guarantees isometry.}
         \label{fig:prop-2}
     \end{subfigure}
     \hfill
     \begin{subfigure}[b]{0.35\textwidth}
         \centering
         \includegraphics[width=\textwidth]{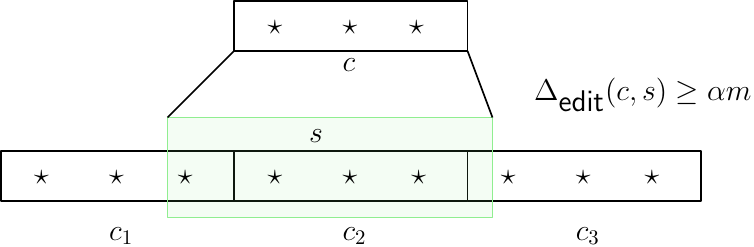}
         \caption{\textbf{Block vs.\ Substring:} The relative edit distance between a block and a substring coming from the concatenation of three blocks is at least $\alpha$.}
         \label{fig:prop-3}
     \end{subfigure}
     \hfill
     \begin{subfigure}[b]{0.38\textwidth}
         \centering
         \includegraphics[width=\textwidth]{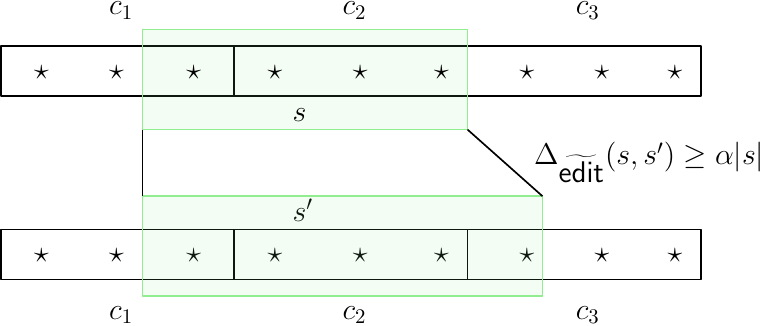}
         \caption{\textbf{Block and a Half vs.\ Substring:} The relative nowhere-vertical edit distance between the concatenation of a block suffix and a full block and its extension by some prefix of a block is at least $\alpha$.}
         \label{fig:prop-5}
     \end{subfigure}
        \caption{Properties of the Misaligner}
        \label{fig:properties}
\end{figure}


\section{Locally Self-Matching Strings: Existence and Construction}

\label{sec:locally-self-matching}

In this section, we define \textit{locally self-matching strings}, which are strings in which every substring has a small relative nowhere-vertical LCS.

\begin{defi}[$\varepsilon$-locally self-matching strings]
\label{def: self-matching}
    For any $\varepsilon\in (0, 1)$, a string $w$ is said to be $\varepsilon$-locally self-matching if for all substrings $s$ of $w$, we have $\slcs(s, s) < \varepsilon|s|$.
\end{defi}
We remark that locally self-matching strings are closely related to synchronization strings introduced by Haeupler and Shahrasbi~\cite{haeupler2017synchronization}. In fact, it is known that if $w$ is an $\varepsilon$-locally self-matching string, then it is also an $\varepsilon'$-synchronization string where $\varepsilon'=2\varepsilon$. The converse is also true: an $\varepsilon$-synchronization string is also an $\varepsilon$-locally self-matching string. Therefore, the notions are equivalent up to constant factors in the parameter $\varepsilon$. For our purposes, it is more convenient to use the formulation in Definition~\ref{def: self-matching} than the more well-known synchronization string formulation. 

It is known that for any $\varepsilon\in (0, 1)$, arbitrarily long $\varepsilon$-locally self-matching strings exist over alphabets of size $\Theta\left(\frac{1}{\varepsilon^2}\right)$~\cite{cheng2018synchronization}. However, to derive explicit bounds on the rate of our embedding, one needs to uncover the constants hiding inside the asymptotic notation. We do so in the following theorem by giving a much tighter analysis of Theorem 4.1 in~\cite{cheng2018synchronization}.   

\begin{theorem}
\label{thm:lll-string}

Let $\Sigma$ be a finite set and $\varepsilon\in \left(0,\frac{1}{2}\right]$ such that the following holds: $$|\Sigma|\ge \frac{e^2}{\varepsilon^2}\cdot \left(1+4\sqrt[4]{\varepsilon}\right).$$
Then, for all positive integers $n$, there exists a $\varepsilon$-locally self-matching string $w$ over $\Sigma$ of length $n$. Moreover, $w$ can be computed in $\textnormal{\textsf{poly}}(n)$ time. 
\end{theorem}
\begin{proof}
    We prove this via the probabilistic method. Let $w$ be a string of length $n$ generated via the following random process.
    
    Let  $\theta:=\frac{e^2}{\varepsilon^{7/4}}$. Randomly pick $\lceil\theta\rceil$ different symbols from $\Sigma$ and let them be the first $\lceil\theta\rceil$ symbols of $w$. If $\lceil\theta\rceil \geq n$, we just pick $n$ different symbols. For $\lceil\theta\rceil + 1 \leq i \leq n$, we pick the $i^{\text{th}}$ symbol $w[i]$ uniformly randomly from $\Sigma \setminus \{w[i-1], \ldots , w[i-\lceil\theta\rceil + 1]\}$.

    We call a substring $s$ of $w$ \textit{bad} if $\slcs(s,s)\geq \varepsilon\cdot  |s|$. Note that $s$ can be bad only if $|s| \geq \lceil\theta\rceil$. Let $ \ell:= |s|$. We have,
    \begin{align*}
        \Pr[s \textnormal{ is bad}] &\leq \binom{\ell}{\varepsilon\ell}^2(|\Sigma|-\lceil\theta\rceil)^{-\varepsilon\ell}\\
        &\leq {\left(\frac{e\ell}{\varepsilon\ell}\right)^{2\varepsilon\ell}(|\Sigma|-\theta)^{-\varepsilon\ell}}\\
        & = C^{-\varepsilon'\ell},
    \end{align*}
where $\varepsilon' :=2\varepsilon$ and $C=\frac{\varepsilon\cdot \sqrt{|\Sigma|-\theta}}{e}\ge   \sqrt{1+3\sqrt[4]{\varepsilon}},$ where the  inequality follows from the assumption on $|\Sigma|$ in the theorem statement.  However,  we have that for any $z\in (0,1)$:
$$(1+z)^2=1+2z+z^2\le 1+2z^{3/4}+z^{3/2}\le 1+3z^{3/4}.$$
We apply the above with $z=\varepsilon^{1/3}$ and by noting that $\varepsilon\in (0,0.5]$,
$$C\ge 1+\sqrt[3]{\varepsilon}.$$

We call $w$ \textit{good} if none of its substrings are bad. 

\begin{lemma}[Claim 1 in \cite{cheng2018synchronization}]
    The badness of two substrings $s$ and $s'$ is mutually independent if $s$ and $s'$ do not intersect, i.e., their intervals of occurrence do not overlap.
\end{lemma}
 By the Lovasz Local Lemma, the probability of $w$ being good is non-zero if for each substring $s$ of $w$, there exists $x_s \in (0, 1)$ such that the following holds:
\[\Pr[s \textnormal{ is bad}] \leq x_s\cdot  \prod_{s' \textnormal{ intersects } s} (1-x_{s'}).\]
We claim that {$x_s := D^{-\varepsilon'|s|}$} works for some choice of $D$ to be determined later. Since the number of length $\ell$ substrings intersecting the substring $s$ is $\ell+|s|$, this amounts to showing that the following inequality is true for every substring $s$ of $w$.
\[C^{-\varepsilon'|s|} \leq D^{-\varepsilon'|s|}\cdot \prod_{\ell=\theta}^n(1-D^{-\varepsilon'\ell})^{|s|+\ell}.\]
This inequality is equivalent to the following inequality for all substrings $s$.
\[C\geq \frac{D}{\prod_{\ell=\theta}^n (1-D^{-\varepsilon'\ell})^{\frac{1+\ell/|s|}{\varepsilon'}}}.\]
The right-hand side of the inequality is maximized when $|s|=\lceil\theta\rceil$. So, it suffices to show that:
\begin{align}C\geq  \frac{D}{\prod_{\ell=\theta}^n(1-D^{-\varepsilon'\ell})^{\frac{\ell+\theta}{\theta\varepsilon'}}}.\label{eq:main}\end{align}

We set $D=1+0.9\cdot\sqrt[3]{\varepsilon}$. To show that \eqref{eq:main} holds, we first prove the following:
\begin{align}\prod_{\ell=\theta}^n(1-D^{-\varepsilon'\ell})^{\frac{\ell+\theta}{\theta\varepsilon'}}\ge  1-\frac{\varepsilon^{7/4}}{e^2}.\label{eq:summain}\end{align}

We will show below that  the following is true:
\begin{align}
    (1-D^{-\varepsilon'\ell})^{\frac{\ell+\theta}{\theta\varepsilon'}}\ge 1-\frac{1}{\ell^2}.
    \label{eq:termmain}
\end{align} 

Let us first see that the above is sufficient to finish the proof.  We first note the following identity (which can be verified by expanding the product):

$$\prod_{i=2}^n\left(1-\frac{1}{i^2}\right)=\frac{n+1}{2n}.$$

Therefore, we have that \eqref{eq:summain} holds:
$$ \prod_{\ell=\theta}^n(1-D^{-\varepsilon'\ell})^{\frac{\ell+\theta}{\theta\varepsilon'}}\ge \prod_{\ell=\theta}^n\left(1-\frac{1}{\ell^2}\right)=\frac{\frac{n+1}{2n}}{\frac{\theta}{2\theta-2}}\ge 1-\frac{1}{\theta}\ge 1-\frac{\varepsilon^{7/4}}{e^2}.$$

We can now verify that \eqref{eq:main} holds:
\begin{align*}
   \frac{1}{C}\cdot \frac{D}{\prod_{\ell=\theta}^n(1-D^{-\varepsilon'\ell})^{\frac{\ell+\theta}{\theta\varepsilon'}}}&\le \frac{1+0.9\cdot\sqrt[3]{\varepsilon}}{1-\frac{\varepsilon^{7/4}}{e^2}}\cdot\frac{1}{1+\sqrt[3]{\varepsilon}}=\frac{1+0.9\cdot\sqrt[3]{\varepsilon}}{1+\sqrt[3]{\varepsilon}-\frac{\varepsilon^{7/4}}{e^2}-\frac{\varepsilon^{25/12}}{e^2}}\\
   &\le \frac{1+0.9\cdot\sqrt[3]{\varepsilon}}{1+\sqrt[3]{\varepsilon}\cdot\left(1-\frac{1}{2^{17/12}\cdot e^2}-\frac{1}{2^{7/4}\cdot e^2}\right)
}<\frac{1+0.9\cdot\sqrt[3]{\varepsilon}}{1+0.909\cdot\sqrt[3]{\varepsilon}}<1,
\end{align*}
where we have used that $\varepsilon\le 0.5$.

Thus, we are left to show \eqref{eq:termmain}, we have:
$$
 \left(1-D^{-\varepsilon'\ell}\right)^{\frac{\ell+\theta}{\theta\varepsilon'}} \ge 1-\frac{\ell+\theta}{\varepsilon'\theta}\cdot D^{-\varepsilon'\ell}.
$$

First notice that the function $f(\ell):=\frac{(\ell+\theta)\ell^2}{\varepsilon'\theta}\cdot D^{-\varepsilon'\ell}$ is decreasing for all $\ell\geq \theta$. To see this note that:
$$
f(\ell+1)=\frac{(\ell+1+\theta)(\ell+1)^2}{\varepsilon'\theta}\cdot D^{-\varepsilon'(\ell+1)}=D^{-\varepsilon'}\cdot \frac{(\ell+1+\theta)(\ell+1)^2}{(\ell+\theta)\ell^2}\cdot f(\ell).
$$

With the goal of relating $f(\ell)$ and $f(\ell+1)$, we compute:
$$
D^{-\varepsilon'}\cdot \left(1+\frac{1}{\ell+\theta}\right)\cdot \left(1+\frac{1}{\ell}\right)^2<D^{-\varepsilon'}\cdot \left(1+\frac{1}{\theta}\right)^3\le D^{-\varepsilon'}\cdot \left(1+\frac{4}{\theta}\right),
$$
where the last inequality follows because $\theta\ge 4$. 
Thus, we have shown that $f(\ell+1)\le f(\ell)\cdot D^{-\varepsilon'}\cdot \left(1+\frac{4}{\theta}\right)$. 
We will now argue that $D^{-\varepsilon'}\cdot \left(1+\frac{4}{\theta}\right)< 1$. 

\begin{claim}\label{claim:increasing}
We have $\left(1+0.9\sqrt[3]{\varepsilon}\right)^{2\varepsilon} > 1 + \frac{4\varepsilon^{7/4}}{e^2}$ for $ \varepsilon \in (0, 0.5] $.
\end{claim}
\begin{proof}
\renewcommand\qedsymbol{$\lrcorner$}
We have from the Taylor series expansion for $ \ln(1+x) $  that for  $ x \in (0,1) $:
\begin{align}
\ln(1 + x) \geq x - \frac{x^2}{2}.    \label{eq:calctool}
\end{align}

Then, we have:
$$
2\varepsilon\cdot \ln\left(1 + 0.9\sqrt[3]{\varepsilon}\right) \geq 2\varepsilon\cdot \left(0.9\sqrt[3]{\varepsilon} - \frac{\left(0.9\sqrt[3]{\varepsilon}\right)^2}{2}\right)=  1.8\cdot\varepsilon^{4/3}-0.81\cdot \varepsilon^{5/3}.
$$

Exponentiating both sides gives us:
\begin{align}
\left(1+0.9\sqrt[3]{\varepsilon}\right)^{2\varepsilon}  \geq e^{1.8\cdot\varepsilon^{4/3}-0.81\cdot \varepsilon^{5/3}}.\label{eq:exprelate}
\end{align}

Again from the Taylor Expansion of $ e^u $ for any $u\in (0,1)$ we have:
$$   e^u \ge 1 + u + \frac{u^2}{2}.
 $$

   Therefore,
\begin{align}
  e^{1.8\cdot\varepsilon^{4/3}-0.81\cdot \varepsilon^{5/3}}\ge 1 + 1.8 \cdot \varepsilon^{4/3} - 0.81 \cdot \varepsilon^{5/3} + 1.62 \cdot \varepsilon^{8/3} - 1.458 \cdot \varepsilon^3 + 0.32805 \cdot \varepsilon^{10/3}.\label{eq:claimderv}
\end{align}

Let $g(\varepsilon):=1.25 \cdot \varepsilon^{4/3} - 0.81 \cdot \varepsilon^{5/3} + 1.62 \cdot \varepsilon^{8/3} - 1.458 \cdot \varepsilon^3 + 0.32805 \cdot \varepsilon^{10/3}$. 
The derivative $ g'(\varepsilon) $ is:
\[
g'(\varepsilon) = \frac{5}{3} \varepsilon^{1/3} - 1.35 \varepsilon^{2/3} + {4.32} \varepsilon^{5/3} - 4.374 \varepsilon^2 + 1.0935\varepsilon^{7/3}
\]
Since 
$1.35 \varepsilon^{2/3}<1.35 \varepsilon^{1/3}$ and $4.374 \varepsilon^2< {4.32} \varepsilon^{5/3}+0.06\varepsilon^{1/3}$ when $\varepsilon\in (0,0.5]$, we have:
\[
g'(\varepsilon) > \left(\frac{5}{3}-1.41\right)\cdot \varepsilon^{1/3}  + 1.0935\cdot \varepsilon^{7/3} >0,
\]
and thus we have $g$ is strictly increasing in $ (0,0.5]$, and $g(0)=0$, we have that $g(\varepsilon)>0$ for all $\varepsilon\in (0,0.5]$. Returning to \eqref{eq:claimderv}, we have:
$$
  e^{1.8\cdot\varepsilon^{4/3}-0.81\cdot \varepsilon^{5/3}}\ge 1+ 0.55\cdot \varepsilon^{4/3} + g(\varepsilon)>1+ 0.55\cdot \varepsilon^{4/3}>1+ 0.55\cdot \varepsilon^{7/4}>1+ \frac{4}{e^2}\cdot \varepsilon^{7/4}.
   $$
   The proof then follows from \eqref{eq:exprelate}.
\end{proof}

Thus, we have proved that $f(\ell+1)<  f(\ell)$. 
Returning to proving \eqref{eq:summain} we have:
$$
 (1-D^{-\varepsilon'\ell})^{\frac{\ell+\theta}{\theta\varepsilon'}}\ge   1-\frac{f(\ell)}{\ell^2}\ge  1-\frac{f(\theta)}{\ell^2}.
$$
The proof concludes now because we claim that $f(\theta)<1$
\begin{claim}
We have $f(\theta)=\frac{2\theta^2}{\varepsilon'}\cdot D^{-\varepsilon'\theta}<1$ for all $\varepsilon\in (0,0.5]$.
\end{claim}
\begin{proof}
\renewcommand\qedsymbol{$\lrcorner$}
    We can rewrite \eqref{eq:calctool} as follows:
    $$(1+x)^{1/x}\ge e^{1-\frac{x}{2}}.$$
Thus, we have that:
$$D^{\varepsilon'\theta}=D^{\frac{2e^2}{\varepsilon^{3/4}}}=\left(\left(1+0.9\sqrt[3]{\varepsilon}\right)^{\frac{10}{9\sqrt[3]{\varepsilon}}}\right)^{\frac{1.8e^2}{\varepsilon^{5/12}}}\ge e^{\left(1-0.45\sqrt[3]{\varepsilon}\right)\cdot{\frac{1.8e^2}{\varepsilon^{5/12}}} }=e^{\frac{1.8e^2}{\varepsilon^{5/12}}-{\frac{0.81\cdot e^2}{\varepsilon^{1/12}}} }>e^{\frac{1.8e^2}{\varepsilon^{5/12}}-{\frac{0.81\cdot e^2}{\varepsilon^{5/12}}} }=e^{\frac{0.99e^2}{\varepsilon^{5/12}}}.$$

On the other hand, we have that 
$$\frac{2\theta^2}{\varepsilon'}=\frac{\theta^2}{\varepsilon}=\frac{e^4}{\varepsilon^{4.5}}<\frac{54.6}{\varepsilon^{4.5}}.$$

Thus proving the claim simply amounts to showing the below holds for all $\varepsilon\in(0,0.5]$:
$$\left(e^{\frac{0.99e^2}{\varepsilon^{5/12}}}\right)^2\ge\left(\frac{54.6}{\varepsilon^{4.5}}\right)^2\iff e^{\frac{1.98e^2}{\varepsilon^{5/12}}}\ge \frac{2982}{\varepsilon^{9}}. $$

The Maclaurin series tells us that:
$$e^{x}=\sum_{k=0}^{\infty}{\frac {x^{k}}{k!}}=1+x+{\frac {x^{2}}{2!}}+{\frac {x^{3}}{3!}}+\cdots$$
By looking at the Maclaurin series and only picking the twenty-third term we have: 
\begin{align*}
   e^{\frac{1.98e^2}{\varepsilon^{5/12}}}> \left(\frac{1.98e^2}{\varepsilon^{5/12}}\right)^{22}\cdot \frac{1}{22!}>\frac{38443}{\varepsilon^{55/6}}>\frac{38443}{\varepsilon^{9}}>\frac{2982}{\varepsilon^9}.&\qedhere
\end{align*}
\end{proof}

The $\textnormal{poly}(n)$-time construction can be done by considering algorithmic versions of the Lovasz Local Lemma. For details, the reader is referred to Lemma 4.2 in~\cite{cheng2018synchronization}. 
\end{proof}

\begin{remark}
\label{rem: three-distinct}
    Note that in the proof, we set $\theta\geq \frac{e^2}{\varepsilon^{7/4}}> 3$. Therefore, every three consecutive symbols in the string whose existence we prove are distinct.
\end{remark}

\section{Embeddings via Misaligners}

\label{sec:embedding}

In this section, we show how to get an embedding using misaligners and locally self-matching strings. On a high level, this embedding takes a locally self-matching string and replaces each of its symbols with elements of some misaligner. This results in a binary string with wildcards. To embed some input string $x$, we then simply instantiate this wildcard string with $x$. Details follow.

Let $\mathcal{C}$ be a 
$(m, k, t, \alpha)$-misaligner and $\Sigma$ be an alphabet such that $|\Sigma|=k$. Fix some arbitrary bijection $\sigma$ from $\Sigma$ to $\mathcal{C}$. Given $\mathcal{C}$ and a positive integer $n$, we compute our embedding map as follows. First, we find the smallest integer $N$ such that $N \geq \frac{tn}{m}$. Next by Theorem~\ref{thm:lll-string}, we compute an $\varepsilon$-locally self-matching string $w$ over $\Sigma$ of length $N$, where $k\geq \frac{e^2}{\varepsilon^2}\cdot \left(1+4\sqrt[4]{\varepsilon}\right)$. Then for every symbol $c$ in $w$, we replace it by $\sigma(c)\in \mathcal{C}$ to obtain the string $w'$. Note that $w'$ has length $mN \geq tn$, and since every string in $\mathcal{C}$ has one wildcard symbol in every $t$ symbols, $w'$ has at least $n$ wildcard symbols. Next, we delete sufficiently many symbols from the right end of $w'$ to ensure that $w'$ has length exactly $tn$ and contains exactly $n$ wildcard symbols. Finally, for every $x\in \{0, 1\}^n$, we define $\varphi_{\mathcal{C}, w, n}(x) := w'_x$,
 i.e., we define $\varphi_{\mathcal{C},w, n}(x)$ to be the string obtained by instantiating $w'$ by $x$.  Clearly, $\varphi_{\mathcal{C}, w, n}: \{0, 1\}^n \to \{0, 1\}^{tn}$ has rate $\frac{1}{t}$. We now make the following claim.

\begin{theorem}
\label{thm:main}
    Let $\mathcal{C}$ be a $(m, k, t, \alpha)$-misaligner and $w$ be an $\varepsilon$-locally self-matching string such that $t\geq \frac{1}{(1-\varepsilon)\alpha -\frac{1}{3m-1}}$. Then for every positive integer $n$ and $x, y \in \{0, 1\}^n$, we have $\ed\left(\varphi_{\mathcal{C}, w, n}(x), \varphi_{\mathcal{C}, w, n}(y)\right) = \ham(x, y)$.
\end{theorem}

\section{Proof of Theorem~\ref{thm:main}}
\label{sec:analysis}

\begin{figure}
    \centering
    \includegraphics{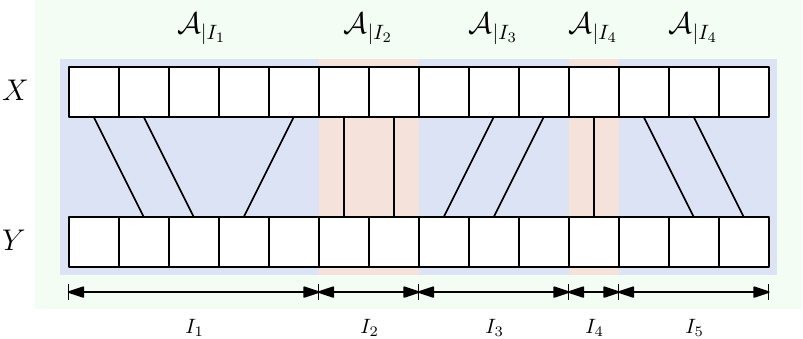}
    \caption{An edit distance alignment $\mathcal{A}$ between the strings $X$ and $Y$. Note the alternating maximal nowhere-vertical and vertical intervals (highlighted blue and red, respectively) $I_1, I_2, I_3, I_4$ and $I_5$ under $\mathcal{A}$. For each $I\in \{I_1, I_2, I_3, I_4, I_5\}$, the induced alignment $\mathcal{A}_{|I}$ is an edit distance alignment between $X_{|I}$ and $Y_{|I}$.}
    \label{fig:alignment-factorization}
\end{figure}
\begin{proof}[Proof of Theorem~\ref{thm:main}]
    Assume for the sake of contradiction that there exist a positive integer $n$ and strings $x, y\in \{0, 1\}^n$ such that $\ed(\varphi_{\mathcal{C}, w, n}(x), \varphi_{\mathcal{C}, w, n}(y)) < \ham(x, y)$. Let $X=\varphi_{\mathcal{C}, w, n}(x)$, $Y=\varphi_{\mathcal{C}, w, n}(y)$ and consider any optimal edit distance alignment $\mathcal{A}$ between $X$ and $Y$. By assumption, we have $\cost(\mathcal{A}) < \ham(X, Y)$. 
    Call a non-empty interval $I\subseteq [1, tn]$ \textit{vertical} under $\mathcal{A}$ if for every $i\in I$, $\mathcal{A}$ aligns $X[i]$ to $Y[i]$. Similarly, call $I$ \textit{nowhere-vertical} under $\mathcal{A}$ if for no $i\in I$, $\mathcal{A}$ aligns $X[i]$ to $Y[i]$\footnote{So, $X[i]$ is either unaligned or aligned to some $Y[j]$ where $i\neq j$.}. Note that an interval could be neither vertical nor nowhere-vertical. The key observation is that $\mathcal{A}$ naturally induces a partition of $[1, tn]$ into alternating maximal vertical and maximal nowhere-vertical intervals, i.e., there exists a unique sequence of intervals $I_1 < I_2< \ldots < I_j$ for some integer $j$ such that $I_1 \cup I_2 \cup \cdots \cup I_j = [1, tn]$, every interval in the sequence is either vertical or nowhere-vertical, and two consecutive intervals are of different types. Additionally, for every $I\in \{I_1, I_2, \ldots , I_j\}$, the alignment $\mathcal{A}$ induces an edit distance alignment between $X_{|I}$ and $Y_{|I}$ (See Figure~\ref{fig:alignment-factorization}). For every $I\in \{I_1, I_2, \ldots , I_j\}$, we refer to the alignment between $X_I$ and $Y_I$ induced by $\mathcal{A}$ as $\mathcal{A}_{|I}$ and its cost as $\cost(\mathcal{A}_{|I})$. Note that   $\cost(\mathcal{A}) = \sum_{I\in \{I_1, I_2, \ldots, I_j\}} \cost(\mathcal{A}_{|I})$. By assumption, we have,
    \[\sum_{I\in \{I_1, I_2, \ldots, I_j\}} \cost(\mathcal{A}_{|I}) = \cost(\mathcal{A}) < \ham(X, Y) = \sum_{I\in \{I_1, I_2, \ldots , I_j\}}\ham(X_{|I}, Y_{|I})\]
    Therefore, there must exist $I \in \{I_1, I_2, \ldots , I_j\}$ such that $\cost(\mathcal{A}_{|I}) < \ham(X_{|I}, Y_{|I})$. We show that this is impossible if the parameters of the misaligner $\mathcal{C}$ and the locally self-matching string $w$ are chosen in a way such that $t\geq \frac{1}{(1-\varepsilon)\alpha -\frac{1}{3m-1}}$ holds.

    \begin{sloppypar}First note that $I$ must be a nowhere-vertical interval under $\mathcal{A}$ since otherwise, we have $\cost(\mathcal{A}_{|I}) = \ham(X_{|I}, Y_{|I})$. Further note that $|I| > 3m -2$ since otherwise, by Property~\ref{prop:short-intervals} of misaligners, $\cost(\mathcal{A}_{|I}) = \selfed(X_{|I}, Y_{|I}) \geq \ed(X_{|I}, Y_{|I}) = \ham(X_{|I}, Y_{|I})$.
    To derive the desired contradiction, we start by lower bounding $\cost(\mathcal{A}_{|I})$. The first step of the lower-bounding process involves \textit{reinstantiating} the wildcard symbols in $X_{|I}$ and $Y_{|I}$ to minimize their nowhere-vertical edit distance. More precisely, let $J\subseteq [|I|]$ be the set of indices in $X_{|I}$ and $Y_{|I}$ that originally contained the symbol $\star$ and were later replaced by symbols of $x$ and $y$ during the embedding. Define $S$ to be the string (with wildcards) obtained by taking $X_{|I}$ (or equivalently $Y_{|I}$) and replacing, for each $j\in J$, $X_{|_I}[j]$ with the symbol $\star$. Then let $x^*, y^*\in \{0, 1\}^{\Gamma(S)}$ be strings minimizing $\selfed(S_{x^*}, S_{y^*})$. Finally, set $X^*:=S_{x^*}$, $Y^*:= S_{y^*}$ and let $\mathcal{A}^*$ be an optimal nowhere-vertical edit distance alignment between $X^*$ and $Y^*$. See Figure~\ref{fig:reinstantiation} for a description of this entire process.
    \end{sloppypar}

    \begin{figure}[t]
        \centering
        \includegraphics[scale=0.6]{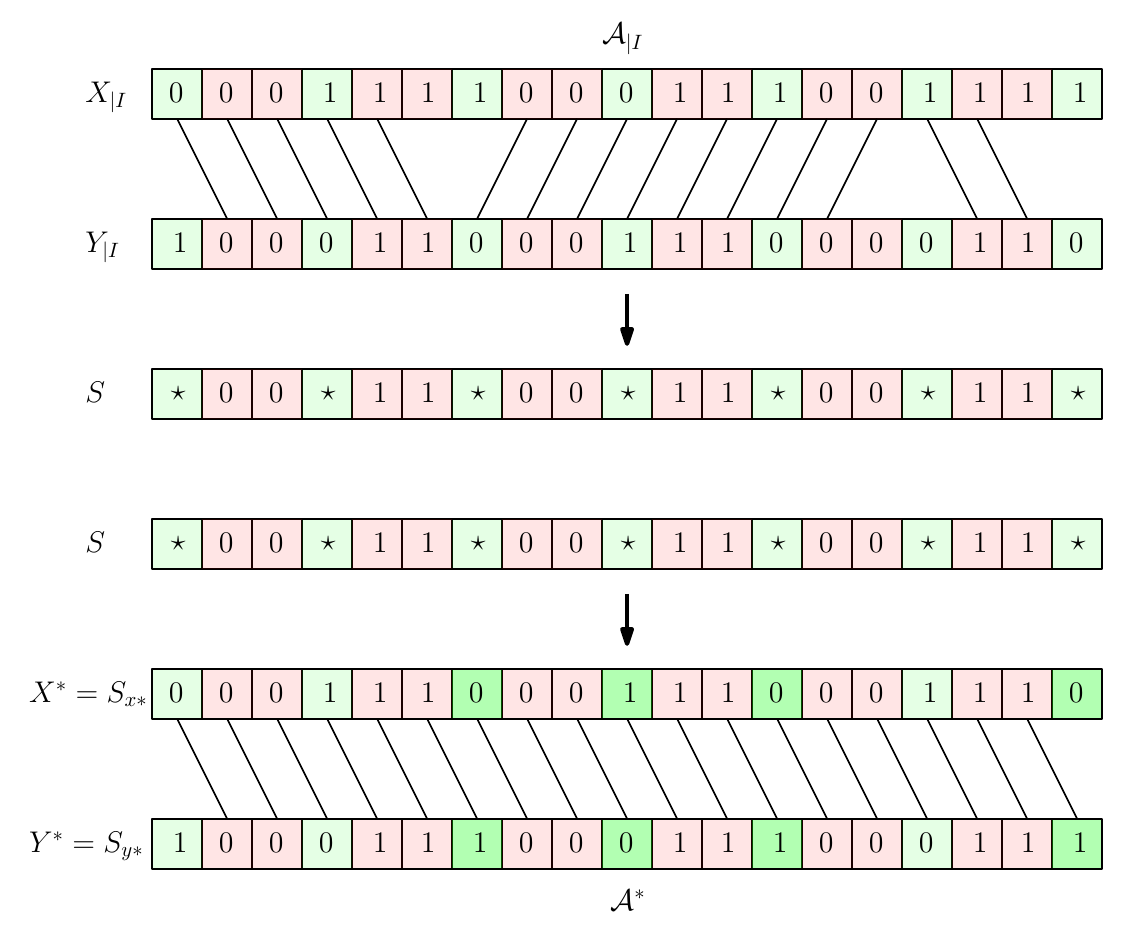}
        \caption{The reinstantiation process. We start with the strings $X_{|I}$ and $Y_{|I}$, replace back the $\star$ symbols where they originally were (marked green) to obtain the string $S$, and reinstantiate $S$ with the strings $x^*$ and $y^*$ to ensure that $\selfed(S_{x^*}, S_{y^*})$ is minimized. Symbols that have changed in the reinstantiation process have been highlighted with a darker shade of green.
        }
        \label{fig:reinstantiation}
    \end{figure}

    Clearly, $\cost_{X_{|I}}^{Y_{|I}}(\mathcal{A}_{|I}) \geq \cost_{X^*}^{Y^*}(\mathcal{A}^*)$. So, by assumption, $\cost_{X^*}^{Y^*}(\mathcal{A}^*) < \ham(X_{|I}, Y_{|I})$ holds.  Call a string a \textit{block} if it is the instantiation of some element in $\mathcal{C}$. Note that both $X^*$ and $Y^*$ start with some proper suffix (possibly empty) of a block, followed by a sequence of blocks, and end in a proper prefix (again, possibly empty). Let $p$ be the number of blocks in $X^*$ (not counting the starting and ending partial blocks). Let $b_1, b_2, \ldots , b_p$ the blocks of $X^*$ in order. Note that since $|I|>3m-2$, $p\geq 2$ and $b_1$ and $b_p$ are distinct. Additionally, let $b_0$ and $b_{p+1}$  be the starting and ending partial blocks, respectively, of $X^*$. Similarly name the blocks of $Y^*$ as $b'_0, b'_1, b'_2, \ldots, b'_p, b'_{p+1}$. The alignment $\mathcal{A}^*$ specifies a sequence of edit operations on $X^*$ transforming it into $Y^*$. This sequence naturally induces on each $b_i$, where $0\leq i \leq p+1$, a sequence of edit operations that transforms it into some string $s_i$,  where we have $s_0\circ s_1 \circ \ldots \circ s_p \circ s_{p+1} = Y^*$ (see Figure~\ref{fig:block-transforms} for an illustration). Note that we have ---
    \[\cost(\mathcal{A}^*) \geq \sum_{i=0}^{p+1}\ed(b_i, s_i)\]
    We now start to lower bound $\cost(\mathcal{A}^*)$. We begin with the following observation.

    \begin{obs}
        \label{obs:bad-blocks}
    
    There exists a string $ Y' $ obtained by re-instantiating some blocks of $ Y^* $, and an alignment $ \mathcal{A}' $ transforming $ X^* $ to $ Y' $, such that the following holds:
    
    \begin{enumerate}
        \item Let $b_i, b'_{i'}$ be instantiations of the same codeword in $\mathcal{C}$, where $i\neq i'$ , then in the alignment $\mathcal{A}'$, if any one pair of characters between $b_i$ and $b'_{i'}$ is matched ``badly'', i.e., for some \(j\in [m]\), \(\mathcal{A}'\) aligns \(b_{i}[j]\) with \(b'_{i'}[j]\), then all the pairs are matched ``badly'', i.e., for \textit{all} \(j\in [m]\), \(\mathcal{A}'\) aligns \(b_{i}[j]\) with \(b'_{i'}[j]\). 
        
        \item The cost of the new alignment does not increase:
        \[
        \cost_{X^*}^{Y'}(\mathcal{A}') \leq \cost_{X^*}^{Y^*}(\mathcal{A}^*).
        \]
    \end{enumerate}
\end{obs}

    \begin{proof}
    \renewcommand{\qedsymbol}{$\lrcorner$}
        We construct sequences $ Y^0, Y^1, \dots, Y^p $ and $ \mathcal{A}^0, \mathcal{A}^1, \dots, \mathcal{A}^p $, starting with $ Y^0 = Y^* $ and $ \mathcal{A}^0 = \mathcal{A}^* $, and ending with $ Y^p = Y' $ and $ \mathcal{A}^p = \mathcal{A}' $. For each block $ b_i $ in $ X^* $ (from $ i = 1 $ to $ p $), we update $ Y^{i} $ and $ \mathcal{A}^{i} $ as follows:

    \begin{enumerate}
        \item If $ b_i $ has no badly matched pairs in $ \mathcal{A}^{i-1} $, set $ Y^{i} = Y^{i-1} $ and $ \mathcal{A}^{i} = \mathcal{A}^{i-1} $.
        \item If $ b_i $ has a badly matched pair with $ b'_{i'} $ (for $ i \neq i' $ and both instantiations of the same $ c \in \mathcal{C} $):
        \begin{enumerate}
            \item \textbf{Re-instantiate $ b'_{i'} $:} Set $ b'_{i'} := b_i $ in $ Y^{i} $.
            \item \textbf{Modify the Alignment:} Update $ \mathcal{A}^{i-1} $ to obtain $\mathcal{A}^{i} $ by:
            \begin{itemize}
                \item Aligning the entire $ b_i $ to $ b'_{i'} $, matching $ b_i[k] $ to $ b'_{i'}[k] $ for all $ k \in [m] $.
                \item Removing any other pair that are in $\mathcal{A}^{i-1}$ and involves $ b_i $ or $ b'_{i'} $ that are not part of this full block alignment.
                \item Keeping all other pairs unchanged.
            \end{itemize}
        \end{enumerate}
    \end{enumerate}

    It is easy to see that, after each round $ \mathcal{A}^{i} $ is still a valid alignment.

    \paragraph{Cost Analysis:}
    
    We need to show that the cost does not increase after each round:
    
    \[
    \cost_{X^*}^{Y^{i}}(\mathcal{A}^{i}) \leq \cost_{X^*}^{Y^{i-1}}(\mathcal{A}^{i-1}).
    \]

    This follows from the following claim.

    \begin{claim}
       Let $ P $ and $ Q $ be strings, and suppose we have an alignment $ \chi $ between $ P $ and $ Q $, where $ P[i] $ is matched to $ Q[j] $, but $ P[i-1] $ is not matched to $ Q[j-1] $. Let $ Q' $ be the string obtained by modifying $ Q[j-1] $ to be equal to $ P[i-1] $ if they are not equal; otherwise, $ Q' = Q $. Then, we can modify $ \chi $ to obtain a new alignment $ \chi' $ between $ P $ and $ Q' $, where $ P[i-1] $ is matched to $ Q'[j-1] $, and the cost does not increase:
    \[
    \cost_{P}^{Q'}(\chi') \leq \cost_{P}^{Q}(\chi).
    \]
    \end{claim}
    
    The claim can be easily proved using some simple case analysis, left as an exercise for the reader.
    
    Applying this claim iteratively for each position $ k $ from $ j $ down to $ 1 $ (and similarly from $ j $ up to $ m $), we can extend the alignment between $ b_i $ and $ b'_{i'} $ to the entire block without increasing the cost.
    
    Therefore, after processing all blocks with bad edges, the total cost satisfies:
    \[
    \cost_{X^*}^{Y'}(\mathcal{A}') \leq \cost_{X^*}^{Y^*}(\mathcal{A}^*).
    \qedhere\]\end{proof}






\begin{figure}
    \centering
    \includegraphics[scale=0.7]{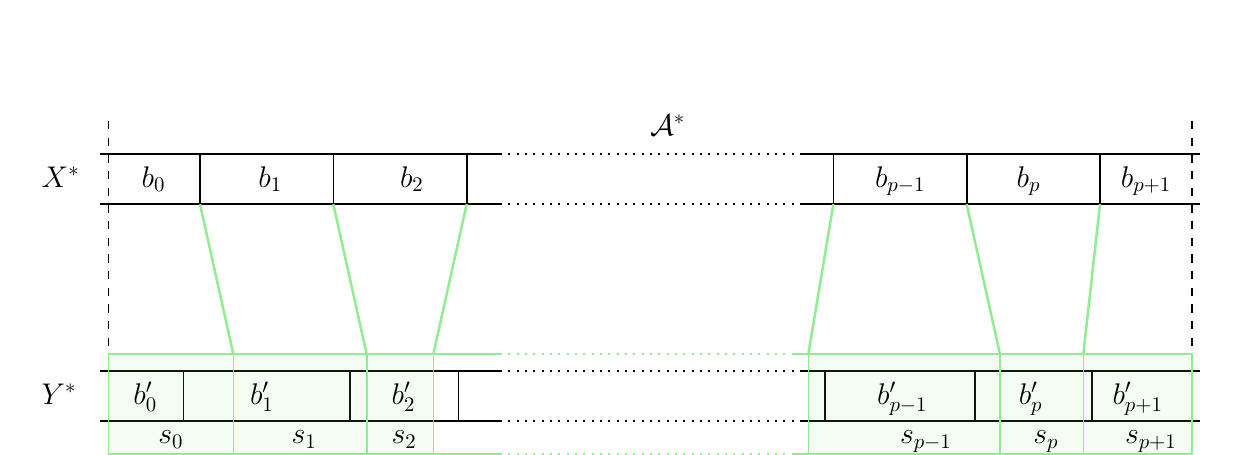}
    \caption{An optimal nowhere-vertical alignment $\mathcal{A}^*$ between $X^*$ and $Y^*$ specifies a sequence of edit operations that transforms each block $b_i$, where $0\leq i \leq p+1$, into a string $s_i$ such that $s_0\circ s_1 \circ \ldots \circ s_p \circ s_{p+1} = Y^*$.}
    \label{fig:block-transforms}
\end{figure}
    For the remainder of the proof, we assume without loss of generality that $Y^*=Y'$, $\mathcal{A}^* = \mathcal{A}'$ where $Y'$ and $\mathcal{A}'$ are the string and the alignment guaranteed to exist by Observation~\ref{obs:bad-blocks}. For each $i\in[p]$, call a block $b_i$ of $X^*$ \textit{bad}, if $\mathcal{A}^*$ aligns $b_i$ to some block $b'_{i'}$ completely and $b_i = b'_{i'}$. From the set of bad blocks $b_i$ and their matches $b'_{i'}$, one can recover a nowhere-vertical common subsequence of a length $p$ substring of $w$. Since $w$ is a $\varepsilon$-locally self-matching string, it follows that the number of bad blocks in $X^*$ is at most $\varepsilon p$. 
    
    Next, call a block $b_i$ of $X^*$ \textit{overworked} if $|s_i| \geq (1+\alpha) m$. Call the rest of the blocks in $X^*$ \textit{good}. We have the following observations.
    \begin{obs}
        If $b_i$ is a bad block, then $\ed(b_i, s_i)=0$.
    \end{obs}
    \begin{proof}
    \renewcommand\qedsymbol{$\lrcorner$}
    This follows from the definition of a bad block. 
    \end{proof}
    \begin{obs}
        If a block $b_i$ is either overworked or good, then $\ed(b_i, s_i) \geq \alpha m$
    \end{obs}
    \begin{proof}
        \renewcommand\qedsymbol{$\lrcorner$}
        If $b_i$ is overworked, then $\ed(b_i, s_i) \geq |s_i| - |b_i| \geq \alpha m$. If $b_i$ is good, then since $|s_i|< (1+\alpha)m < 2m$, $s_i$ can intersect at most three consecutive blocks in $Y^*$. Let us call these blocks $b'_{i'-1}, b'_{i'}$ and $b'_{i'+1}$ for some $i'\in [p]$. By Remark~\ref{rem: three-distinct}, all three of these blocks are instantiations of distinct elements from $\mathcal{C}$. Now let $c_i, c_{i'-1}, c_{i'}, c_{i'+1}\in\mathcal{C}$ be strings such that $b_i, b'_{i'-1}, b'_{i'}, b'_{i'+1}$ are instantiations of $c_i, c_{i'-1}, c_{i'}, c_{i'+1}$, respectively.  If none of $c_{i'-1}, c_{i'}, c_{i'+1}$ are equal to $c_i$, then by Property~\ref{prop: full-block-distance-1} of misaligners, we have $\ed(b_i, s_i) \geq \alpha m$. Furthermore, if exactly one, say $c_{i'-1}$ is equal to $c_i$, then there does not exist $j\in[m]$ such that $\mathcal{A}^*$ aligns $b_i[j]$ to $b'_{i'-1}[j]$ since otherwise, by Observation~\ref{obs:bad-blocks} and our assumption on $\mathcal{A}^*$, $b_i$ would be a bad block. Therefore, we can apply Property~\ref{prop: full-block-distance-2} to get $\ed(b_i, s_i) \geq \alpha m$.
    \end{proof}
We also have the following observations regarding the first partial block $b_0$.
\begin{obs}
    If $b_1$ is a bad block, then $\ed(b_0, s_0) \geq \alpha |b_0|$
\end{obs}
\begin{proof}
\renewcommand\qedsymbol{$\lrcorner$}
    If $b_1$ is a bad block, then $b'_{0}\circ b'_1$ is a prefix of $s_0$. Since $\alpha<1$, $|s_0| \geq |b'_{0}\circ b'_1| \geq (1+\alpha)|b_0|$. So, $\ed(b_0, s_0) \geq |s_0|-|b_0| \geq \alpha b_0$.
\end{proof}
\begin{obs}
    If $b_1$ is not a bad block, then $\ed(b_0\circ b_1, s_0\circ s_1) \geq \alpha(|b_0|+|b_1|)$.
\end{obs}
\begin{proof}
\renewcommand\qedsymbol{$\lrcorner$}
    If $|s_0\circ s_1| \geq (1+\alpha)(|b_0|+ |b_1|)$, then we are done. If $|s_0\circ s_1| < (1+\alpha)(|b_0|+ |b_1|)$, then since $\alpha\leq \frac{1}{2}$, $s_0\circ s_1$ must be a prefix of $b'_{0}\circ b'_{1}\circ b'_2$. By Property~\ref{prop: boundary-blocks-1} of misaligners, it then follows that $\ed(b_0\circ b_1, s_0\circ s_1) \geq \alpha(|b_0|+|b_1|)$. 
\end{proof}
Similar statements for the final partial block $b_{p+1}$ also hold with essentially the same proofs.
\begin{obs}
    If $b_p$ is a bad block, then $\ed(b_{p+1}, s_{p+1}) \geq \alpha |b_{p+1}|$
\end{obs}
\begin{obs}
    If $b_p$ is not a bad block, then $\ed(b_p\circ b_{p+1}, s_p\circ s_{p+1}) \geq \alpha(|b_p|+|b_{p+1}|)$.
\end{obs}
We can now give a lower bound for $\cost(\mathcal{A}^*)$. We claim the following.
\begin{claim}
$\cost(\mathcal{A}^*) \geq (1-\varepsilon)\alpha |I|$
\end{claim}
\begin{proof}
    \renewcommand\qedsymbol{$\lrcorner$}
    We break the proof into four cases depending on whether or not $b_1$ and $b_{p}$ are bad.
    
    \begin{itemize}
        \item \textbf{Case 1: Both $b_1$ and $b_p$ are bad.}

        In this case, we have:\allowdisplaybreaks
        \begin{align*}
            \cost(\mathcal{A}^*) &\geq \sum_{i=0}^{p+1}\ed(b_i, s_i)\\
            &= \ed(b_0, s_0) + \ed(b_{p+1}, s_{p+1}) + \sum_{i=1}^{p}\ed(b_i, s_i) \\
            &\geq \alpha|b_0| + \alpha|b_{p+1}| + \left(\sum_{\substack{{b_i \textnormal{ is bad}}\\{1\leq i \leq p}}}\ed(b_i, s_i) + \sum_{\substack{{b_i \textnormal{ is not bad}}\\{1\leq i \leq p}}}\ed(b_i, s_i)\right) \\
            &\geq \alpha(|b_0|+|b_{p+1}|) + (0 + (1-\varepsilon)p\alpha m)\\
            &\geq (1-\varepsilon)\alpha(|b_0|+|b_{p+1}|) + (1-\varepsilon)p\alpha m\\
            &= (1-\varepsilon)\alpha(|b_0|+|b_{p+1}|+pm)\\
            &= (1-\varepsilon)\alpha|I|
        \end{align*}
        \item \textbf{Case 2: $b_p$ is bad but $b_1$ is not.}

        In this case, we have:
        \begin{align*}
            \cost(\mathcal{A}^*) &\geq \sum_{i=0}^{p+1}\ed(b_i, s_i)\\
            &=\ed(b_0\circ b_1, s_0\circ s_1) + \ed(b_{p+1}, s_{p+1}) + \sum_{i=2}^{p}\ed(b_i, s_i)\\
            &\geq \alpha(|b_0|+|b_1|) + \alpha|b_{p+1}| + \left(\sum_{\substack{{b_i \textnormal{ is bad}}\\{2\leq i \leq p}}}\ed(b_i, s_i) + \sum_{\substack{{b_i \textnormal{ is not bad}}\\{2\leq i \leq p}}}\ed(b_i, s_i)\right) \\
            &\geq \alpha(|b_0|+b_{p+1}) + \alpha |b_1| + (0+((1-\varepsilon)p-1)\alpha m)\\
            & = \alpha(|b_0|+b_{p+1}) + (1-\varepsilon)p\alpha m\\
            &\geq (1-\varepsilon)\alpha(|b_0|+|b_{p+1}|) + (1-\varepsilon)p\alpha m\\
            &= (1-\varepsilon)\alpha(|b_0|+|b_{p+1}|+pm)\\
            &= (1-\varepsilon)\alpha|I|
        \end{align*}
        \item \textbf{Case 3: $b_1$ is bad but $b_p$ is not.}
        
        This case is essentially identical to Case 2.
        \item \textbf{Case 4: Neither $b_1$ nor $b_p$ are bad.}

        Again, we have:
        \begin{align*}
            \cost(\mathcal{A}^*) &\geq \sum_{i=0}^{p+1}\ed(b_i, s_i)\\
            &=\ed(b_0\circ b_1, s_0\circ s_1)+\ed(b_p\circ b_{p+1}, s_p\circ s_{p+1})+ \sum_{i=2}^{p-1}\ed(b_i, s_i)\\
            &\geq \alpha(|b_0|+|b_1|) +\alpha(|b_p|+|b_{p+1}|)+ \left(\sum_{\substack{{b_i \textnormal{ is bad}}\\{2\leq i \leq p-1}}}\ed(b_i, s_i) + \sum_{\substack{{b_i \textnormal{ is not bad}}\\{2\leq i \leq p-1}}}\ed(b_i, s_i)\right)\\
            &\geq \alpha(|b_0|+|b_1|) +\alpha(|b_p|+|b_{p+1}|) + (0+((1-\varepsilon)p-2)\alpha m)\\
            & = \alpha(|b_0|+|b_{p+1}|) + (1-\varepsilon)p\alpha m\\
            &\geq (1-\varepsilon)\alpha(|b_0|+|b_{p+1}|) + (1-\varepsilon)p\alpha m\\
            &= (1-\varepsilon)\alpha(|b_0|+|b_{p+1}|+pm)\\
            &= (1-\varepsilon)\alpha|I|
        \end{align*}
    \end{itemize}
\end{proof}
Now recall that we had $t\geq \frac{1}{(1-\varepsilon)\alpha -\frac{1}{3m-1}}$. Since $|I|\geq 3m-1$, this means $t\geq \frac{1}{(1-\varepsilon)\alpha -\frac{1}{|I|}}$. Multiplying and rearranging, we get:
\[\cost(\mathcal{A}^*) \geq (1-\varepsilon)\alpha |I| \geq \frac{|I|}{t} +1\]
Since $\frac{|I|}{t}+1$ is an upper bound on $\ham(X_{|I}, Y_{|I})$, $\cost(\mathcal{A}^*) < \ham(X_{|I}, Y_{|I})$ cannot hold, and hence $\cost(\mathcal{A}_{|I}) < \ham(X_{|I}, Y_{|I})$ cannot hold --- a contradiction!
\end{proof}

\section{An Explicit Constant Rate Embedding for Binary Strings}
\label{sec:explicit-const}

In this section, we give an isometric embedding of the Hamming metric into the edit metric with rate $\frac{1}{8}$, thus proving Theorem~\ref{thm:explicit}. The first step of the proof is showing the existence of a misaligner with a specific set of parameters.

\begin{lemma}
\label{lemma:misaligner-exists}
    A $(320, 676, 8, 0.1625)$-misaligner exists. 
\end{lemma}
\begin{proof}
We show the existence of a $(320, 676, 8, 0.1625)$-misaligner via a computer search. The elements of the misaligner and the code used to find them can be seen in \cite{Code}. A high-level description of the code can also be found in Appendix~\ref{sec:misaligner-code-description}.  
\end{proof}


By Theorem~\ref{thm:main}, in order to obtain a rate $\frac{1}{8}$-isometric embedding from a $(320, 676, 8, 0.1625)$-misaligner, it suffices to have a $0.224$-locally self-matching string over an alphabet of size 676. Theorem~\ref{thm:lll-string} guarantees the existence of a $0.224$-locally self-matching string over an alphabet of size only 553. Thus, we are done with room to spare!

\begin{theorem}
    \label{thm:rate-1/8}
    For every positive integer $n$, there exists an isometric embedding $\varphi_n: \{0, 1\}^n \to \{0, 1\}^{8n}$ of the Hamming metric into the edit metric.
\end{theorem}

\section{Embeddings for Strings over Larger Alphabets}

\label{sec:close-to-1/3}

So far, we have focused entirely on finding the best possible rate of an isometric embedding that maps binary strings to binary strings. However, it is natural to ask what happens if we consider larger alphabets. Let $\Sigma$ be an arbitrary alphabet, and consider an isometric embedding $\varphi:\Sigma^n \to \Sigma^N$ of the Hamming metric into the edit metric.  What is the highest rate $\varphi$ could possibly have?

As in the binary case, we can try to tackle this question by using misaligners for larger alphabets. In fact, it is not hard to see that misaligners designed for the binary alphabet also function as misaligners for larger alphabets, without any loss in parameters. Therefore, increasing the alphabet size can only improve the achievable rate.

In this section, we show that even without relying on the full machinery of misaligners, one can obtain isometric embeddings with rates arbitrarily close to $\frac{1}{3}$ by making the alphabet large enough.

\begin{theorem}
\label{thm:rate-close-to-1/3}
    For any $\rho>0$, there exists an alphabet $\Sigma$ and a family of functions $\varphi_n:\Sigma^n\to\Sigma^N$ for each positive integer $n$ such that the following hold.
    \begin{enumerate}
        \item For every $n$, $\varphi_n$ is an isometric embedding of the Hamming metric into the edit metric, i.e., for all $x, y\in \Sigma^n$, we have $\ed\left(\varphi_n(x), \varphi_n(y)\right) = \ham(x, y)$.
        \item The rate of $\varphi_n$ is at least $\frac{1}{3}-\rho$, i.e., $\frac{n}{N}\geq \frac{1}{3}-\rho$.
    \end{enumerate}
\end{theorem}

\begin{proof}
Fix any $\rho>0$ and define $R:= \frac{1}{3}-\rho$. Let $\varepsilon := \frac{1-3R}{1-R}$ and $\varepsilon'=\frac{\varepsilon}{2}$. We will choose $\Sigma=\left[\left\lceil\frac{5e^2}{\varepsilon'^2}\right\rceil\right]$ as our alphabet.

Fix any positive integer $n$. Our embedding map $\varphi_n:\Sigma^n\to \Sigma^N$ is defined as follows. First, note that we may assume without loss of generality that $R$ is rational, so let $R=\frac{a}{b}$ for coprime positive integers $a$ and $b$. Define $q:= \lfloor\frac{b}{a}\rfloor$, $r:= b-aq$, and $m:= \lfloor\frac{n}{a}\rfloor$. 

Next, fix any $\varepsilon'$-locally self-matching string $w$ of length $n':=m(q-1+r)+(n-m)(q-1)$ over $\Sigma$. Since $R\leq \frac{1}{3}$, we have $\varepsilon'\leq \frac{1}{2}$ and by Theorem~\ref{thm:lll-string}, $\Sigma$ is large enough for such a string $w$ to always exist. 

Now, partition $w$ into $n$ contiguous substrings $w^{(1)}, w^{(2)}, \ldots , w^{(n)}$ so that every $a^{\text{th}}$ substring in this sequence has length $(q-1+r)$ while all other substrings have length $(q-1)$. Since $m=\lfloor\frac{n}{a}\rfloor$, exactly $m$ of the $w^{(i)}$'s have length $(q-1+r)$ while the remaining have length $(q-1)$. Thus, the total length of the $w^{(i)}$'s indeed sums to $m(q-1+r)+(n-m)(q-1)=n'$. 

Finally, given any $x\in \Sigma^n$, define $\varphi_n(x)\in\Sigma^N$ as follows.
\[\varphi_n(x)=x[1]\circ w^{(1)}\circ x[2]\circ w^{(2)}\circ \cdots x[n]\circ w^{(n)}\]

Let us first show that the rate of $\varphi_n$, which is given by the expression $\frac{n}{n+n'}$ is indeed at least $R$.
\allowdisplaybreaks
\begin{align*}
    \frac{n}{n+n'} & = \frac{n}{n+m(q-1+r)+(n-m)(q-1)}\\
    &=\frac{n}{m(q+r)+(n-m)q}\\
    &= \frac{n}{mr+nq}\\
    &= \frac{n}{\lfloor \frac{n}{a}\rfloor\cdot r + nq}\\
    &\geq \frac{n}{\frac{n}{a}\cdot r + nq}\\
    & = \frac{a}{r+aq}\\
    & = \frac{a}{b}\\
    & = R
\end{align*}
In fact, it is not hard to see that not only is the rate $R$ globally but also locally, i.e., for any substring of the output string, roughly an $R$-fraction of the symbols come from the input string. We will need the following observation later in our analysis.

\begin{obs}
    \label{obs:input-density}
    For all $x\in\Sigma^n$ and for all substrings $w'$ of $\varphi_n(x)$, the number of symbols in $w'$ coming from $x$ is at most $R|w'|+1$ and at least \(R|w'|-1\).
\end{obs}

Now, we show that $\varphi_n$ is indeed isometric for every positive integer $n$. Suppose not; then there exists a positive integer $n$ and $x, y\in \Sigma^n$ such that $\ed(\varphi_n(x), \varphi_n(y))<\ham(x, y)$. Let $X:=\varphi_n(x), Y:=\varphi_n(y)$, and $\mathcal{A}$ be any optimal edit distance alignment between $X$ and $Y$. Since $\ham(x, y) = \ham(X, Y)$, by assumption, we have $\cost(\mathcal{A}) < \ham(X, Y)$. Similar to the proof of Theorem~\ref{thm:main}, consider the partition of $[N]$ into alternating maximal vertical and maximal nowhere-vertical intervals under $\mathcal{A}$. By the same argument as in the proof of Theorem~\ref{thm:main}, there must exist a nowhere-vertical interval $I$ in this partition such that $\cost(\mathcal{A}_{|I}) < \ham(X_{|I}, Y_{|I})$. 

Call a symbol in $X$ or $Y$ \textit{frozen} if it comes from the locally self-matching string $w$, and \textit{mutable}, otherwise. We can assume without loss of generality that \(X_{|I}\) and \(Y_{|I}\) both start with mutable symbols; otherwise, we can modify \(\mathcal{A}\) by aligning the first symbols of \(X_{|I}\) and \(Y_{|I}\), producing a new alignment that can only be cheaper. By the same argument, we may assume that \(X_{|I}\) and \(Y_{|I}\) also end in mutable symbols.

Finally, let \(s\) be the string obtained from \(X_{|I}\) by removing all mutable symbols from \(X_{|I}\).  We will now show that $\cost(\mathcal{A}_{|I})<\ham(X_{|I}, Y_{|I})$ cannot hold by considering two cases based on the length of $I$.
\begin{itemize}
    \item \textbf{Case 1: $|I|< \frac{1}{1-R}\left(\frac{8}{\varepsilon}+1\right)$.}

   Recall from the proof of Theorem~\ref{thm:lll-string} that our construction of the $\frac{\varepsilon}{2}$-locally self-matching string \(w\) guarantees that any
\(\theta := \frac{e^2}{(\varepsilon/2)^{7/4}}\) consecutive symbols of \(w\) are pairwise distinct. Moreover, since
\(
|I| < \frac{1}{1-R}\left(\frac{8}{\varepsilon}+1\right) \le \theta
\)
for all \(\varepsilon\in [0, 1]\), all frozen symbols in \(X_{|I}\) are distinct. Thus, \(\slcs(s, s)=0\). Let \(k\) be the number of mutable symbols in \(X_{|I}\) (or equivalently \(Y_{|I}\)). We analyze two subcases.
    \begin{itemize}
        \item \textbf{Case 1a: \(|I|\) is a multiple of 3.}
        
        In this case, \(k\leq \frac{|I|}{3}\) since in every three consecutive symbols in \(X_{|I}\), at most one is mutable. Given \(\cost(\mathcal{A}_{|I}) <\ham(X_{|I}, Y_{|I}) \leq k\leq\frac{|I|}{3}\), we deduce that \(\slcs(X_{|I}, Y_{|I}) > \frac{2|I|}{3}\). However, this would mean \(\slcs(s, s)\geq \slcs(X_{|I}, Y_{|I}) - 2k >0\), which is a contradiction.

        \item  \textbf{Case 1b: \(|I|\) is not a multiple of 3.}

        \begin{sloppypar}
            
        In this case, \(k\leq \lfloor\frac{|I|}{3}\rfloor+1\). Once again, by assumption, we have \(\cost(\mathcal{A}_{|I}) <\ham(X_{|I}, Y_{|I}) \leq k\leq\left\lfloor\frac{|I|}{3}\right\rfloor+1\). So, \(\slcs(X_{|I}, Y_{|I}) > |I|-\left\lfloor\frac{|I|}{3}\right\rfloor-1\). We now make use of the fact that both the first and last symbols of \(X_{|I}\) and \(Y_{|I}\) are mutable. Because of this, the first symbols of \(X_{|I}\) and \(Y_{|I}\) cannot both participate in a nowhere-vertical common subsequence. The same is true for the last symbols of \(X_{|I}\) and \(Y_{|I}\). Thus, \(\slcs(s, s)\geq \slcs(X_{|I}, Y_{|I}) - 2k -2 >|I|-3\lfloor\frac{|I|}{3}\rfloor-1\geq 0\), again yielding a contradiction.
        \end{sloppypar}
        
    \end{itemize}

\item \textbf{Case 2: $|I|\geq \frac{1}{1-R}\left(\frac{8}{\varepsilon}+1\right)$.}

By Observation~\ref{obs:input-density}, $(1-R)|I|+1\geq |s| \geq (1-R)|I|-1$. Since $s$ is a substring of a $\frac{\varepsilon}{2}$-locally self-matching string, we must have
    \(
    \slcs(s, s) \leq \frac{\varepsilon}{2} |s|.
    \)
    Furthermore, since $|I| \geq \frac{1}{1-R}\left(\frac{8}{\varepsilon}+1\right)$, we have $|s| \geq (1-R)|I|-1 \geq \frac{8}{\varepsilon}$. And for any \(|s|\geq \frac{8}{\varepsilon}\), we must have \(\frac{\varepsilon}{2}|s| \leq \varepsilon|s|-4\). Consequently,
    \begin{align}
    \slcs(s, s) \leq \frac{\varepsilon}{2} |s| \leq \varepsilon|s|-4 \leq \varepsilon((1-R)|I|+1)-4 \leq \varepsilon(1-R)|I| - 3 \label{ineq:slcs-ub}
    \end{align}
    Meanwhile, since $\cost(\mathcal{A}_{|I}) < \ham(X_{|I}, Y_{|I}) \leq R|I| + 1$, it follows that
    \(
    \slcs(X_{|I}, Y_{|I}) > (1-R)|I| - 1.
    \)
    So, by the same argument as before,
    \begin{align}
    \slcs(s, s) > (1-R)|I| + 1 - 2(R|I| + 1) > (1-3R)|I| - 3 \label{ineq:slcs-lb}
    \end{align}
    Combining inequalities (\ref{ineq:slcs-ub}) and (\ref{ineq:slcs-lb}) yields:
    \[
    \varepsilon(1-R)|I| - 3 > (1-3R)|I| - 3,
    \]
    which simplifies to $\varepsilon > \frac{1-3R}{1-R}$, a contradiction.

\end{itemize}

\end{proof}

\begin{remark}
    Although Theorem~\ref{thm:rate-close-to-1/3} gives isometric embeddings with rates arbitrarily close to \(\frac{1}{3}\), these embeddings do not have the best rate vs. alphabet size tradeoff. In particular, Theorem~\ref{thm:rate-close-to-1/3} gives a rate \(\frac{1}{8}\) embedding over an alphabet of size 290. In contrast, Theorem~\ref{thm:rate-1/8} achieves the same rate using only a binary alphabet. Therefore, for a better rate-to-alphabet size tradeoff, one should opt for misaligner-based embeddings.
\end{remark}

\section{Embeddings without Interleaving?}
\label{sec:isometry-implies-interleaving}
For now, let us again return to considering only binary strings. Our embedding from Section~\ref{sec:embedding} is an example of what we might call an \textit{interleaved embedding}, where the output string is constructed by interleaving the input bits with sequences of bit strings that do not depend on the input. One might ask whether there exist isometric embeddings, potentially with better rates, that do not follow this interleaving framework. In this section, we show that the answer is no. In particular, we show that \textit{every} isometric embedding of the Hamming metric into the edit metric must be, in a sense, an interleaved embedding.

To prove this claim, one must formally define interleaved embeddings. Our definition is going to naturally extend the intuitive notion of interleaving input bits with fixed bit patterns by allowing two additional flexibilities. First, we will allow the input bits to appear \textit{out-of-order} --- e.g., the second input bit might appear after the first input bit in the output string. Second, we will allow some of the input bits to appear \textit{complemented} in the output string.
\begin{defi}[Interleaved Embedding]
\label{def:interleaved-embedding}
    A function $\varphi:\{0, 1\}^n \to \{0, 1\}^N$ with $N\geq n$ is called an {interleaved embedding} if there exist ---
    \begin{itemize}
        \item an injective function $\eta : [n] \to [N]$,
        \item a collection of functions $\pi_1, \pi_2, \ldots , \pi_n: \{0, 1\} \to \{0, 1\}$ such that for each $i\in[n]$, $\pi_i(b) =b $ or $\pi_i(x)=1-{b}$ for all $b\in \{0, 1\}$, and
        \item a fixed string $w\in \{0, 1\}^{N-n}$,
    \end{itemize}
    such that for all $x\in \{0, 1\}^n$, if $X=\varphi(x)$, then the following hold.
    \begin{itemize}
        \item For all $i\in [n]$, $X[\eta(i)]= \pi_i(x[i])$.
        \item Let $S\subseteq [N]$ be the set of indices that are not in the image of $\eta$, i.e., $S=\{j\in [N]: \nexists i \textnormal{ with } \eta(i) = j\}$. Then $X_{|S}=w$.
    \end{itemize}

\end{defi}
Definition~\ref{def:interleaved-embedding} captures all functions that interleave the input bits with sequences of fixed bit strings, following a Hamming distance preserving preprocessing of the input space. This preprocessing involves applying a fixed permutation to the input bits and flipping a chosen subset of them. We now show that this framework is expressive enough to represent all isometric embeddings of the Hamming metric into the edit metric.

\begin{theorem}
    \label{thm:isometry-implies-interleaving}
    If $\varphi: \{0, 1\}^n \to \{0, 1\}^N$ is an isometric embedding of the Hamming metric into the edit metric, then $\varphi$ is necessarily an interleaved embedding.
\end{theorem}
\begin{proof}
    Fix any arbitrary isometric embedding $\varphi:\{0, 1\}^n \to \{0, 1\}^N$ of the Hamming metric into the edit metric. We begin by setting up some notation. For each $x\in \{0, 1\}^n$ and $i\in [n]$, let us denote by $x^{\oplus i}$ the bit string obtained by flipping the $i^{\textnormal{th}}$ bit in $x$. Now fix any $x\in \{0, 1\}^n$ and $i\in [n]$. Since $\varphi$ is an isometric embedding of the Hamming metric into the edit metric, we have $\ed(\varphi(x), \varphi(x^{\oplus i})) = \ham(x, x^{\oplus i}) = 1$. It follows that there exists an index $k\in [N]$ such that $\varphi(x^{\oplus i}) = \varphi(x)^{\oplus k}$. We will denote this index $k$ by $\eta_{x}(i)$. In other words, $\eta_{x}(i)$ is the unique index at which we need to perform a bit-flip in $\varphi(x)$ to obtain $\varphi(x^{\oplus i})$. We first make the observation that flipping distinct bits in the input always affects distinct bits in the output, i.e., the function $\eta_x:[n]\to [N]$ is injective for all $x\in \{0, 1\}^n$.
    \begin{obs}
        \label{obs:different-direction-index}
        Let $i, j\in [n]$ such that $i\neq j$. Then for all $x\in \{0, 1\}^n$, we have $\eta_x(i) \neq \eta_x(j)$.
    \end{obs}
    \begin{proof}
     \renewcommand\qedsymbol{$\lrcorner$}
     Assume for the sake of contradiction that for some $x\in \{0, 1\}^n$, $\eta_x(i) = \eta_x(j)$. Then, we have, $2=\ham(x^{\oplus i}, x^{\oplus j}) = \ed(\varphi(x^{\oplus i}), \varphi(x^{\oplus j}))=\ed(\varphi(x)^{\oplus \eta_x(i)}, \varphi(x)^{\oplus \eta_x(j)}) = 0$, which is a contradiction. 
    \end{proof}
    The following observation is also almost immediate.
    \begin{obs}
        \label{obs:reflection-invariance}
        For every $i\in [n]$ and $x\in \{0, 1\}^n$, we have $\eta_x(i) = \eta_{x^{\oplus i}}(i)$.
    \end{obs}
    \begin{proof}
        \renewcommand\qedsymbol{$\lrcorner$}
        Fix \(x\in \{0, 1\}^n\), \(i\in [n]\) and write \(\alpha := \eta_x(i)\),  \(\beta:= \eta_{x^{\oplus i}}(i)\). Then we have ---
        \[\varphi(x) = \varphi\left(\left(x^{\oplus i}\right)^{\oplus i}\right) = \varphi\left(x^{\oplus i}\right)^{\oplus \beta}=\left(\varphi(x)^{\oplus \alpha}\right)^{\oplus \beta},\]
        which forces \(\alpha =\beta\).
    \end{proof}
    The key insight in our proof is the observation that flipping the $i^{\textnormal{th}}$ bit of any input string for some fixed $i$ always affects the same bit in the output and is independent of the input string.
    \begin{lemma}
        \label{lemma:same-direction-index}
        Fix any $i\in [n]$ and $x\in \{0, 1\}^n$. Then for every $y\in \{0, 1\}^n$, we have $\eta_x(i) =\eta_y(i)$.
    \end{lemma}
    \begin{proof}
        \renewcommand{\qedsymbol}{$\lrcorner$}
        First, note that it suffices to prove the claim for all $y\in \{0, 1\}^n$ such that $y[i]=x[i]$ since by Observation~\ref{obs:reflection-invariance}, $\eta_y(i)=\eta_{y^{\oplus i}}(i)$ for every $y$. Therefore, in what follows, $y$ will always be a string such that $y[i]=x[i]$.

        Our proof is going to be by induction on the Hamming distance of $y$ from $x$. Clearly, the claim is true if $\ham(x, y)=0$. For the inductive hypothesis, assume that the claim is true for all $y\in \{0, 1\}^n$ such that $0\leq \ham(x, y) < d$. For the inductive step, let $y\in \{0, 1\}^n$ be a string such that $\ham(x, y)=d$. Then there exists $z\in \{0, 1\}^n$ such that $\ham(x, z)=d-1$ and $\ham(z, y)=1$. Indeed, we can choose $z$ to be the second-to-last vertex in any shortest path from $x$ to $y$ in the Boolean hypercube. By the inductive hypothesis, we have $\eta_x(i)=\eta_{z}(i)$. Therefore, it suffices to prove that $\eta_{z}(i)=\eta_y(i)$.

        \begin{figure}
            \centering
            \includegraphics[scale=1]{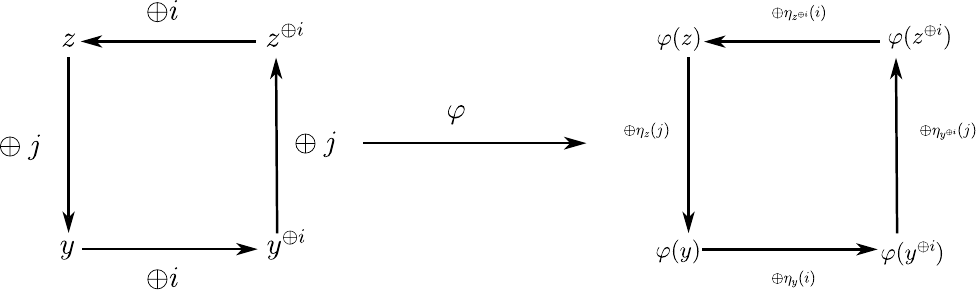}
            \caption{A sequence of bit-flips that both starts and ends at $z$, and their corresponding effects on the image strings under $\varphi$. Any two consecutive arrows correspond to bit flips at distinct locations. Therefore, in order to cancel out all bit flips, non-consecutive arrows must correspond to bit flips at the same locations.}
            \label{fig:cycle}
        \end{figure}

        Assume otherwise. Since $x[i]=y[i]$ and $z$ is on a shortest path from $x$ to $y$, $z[i]=y[i]$. Then we must have $y=z^{\oplus j}$ for some $j\neq i$. Consider walking along the following cycle in the Boolean hypercube and observing the images of the encountered strings under $\varphi$ (see Figure~\ref{fig:cycle} for an illustration) ---
        \[z\xrightarrow{\phantom{aa} \oplus j\phantom{aaa}} y \xrightarrow{\phantom{aa} \oplus i \phantom{aaa}} y^{\oplus i} \xrightarrow{\phantom{aa} \oplus j \phantom{aaa}} z^{\oplus i} \xrightarrow{\phantom{aa} \oplus i \phantom{aaa}} z\]
        The images of these strings under $\varphi$ follow a corresponding cycle given by ---
        \begin{equation*}
            \varphi(z) \xrightarrow{\phantom{aa} \oplus \eta_z(j)\phantom{aaa}} \varphi(y) \xrightarrow{\phantom{aa} \oplus \eta_y(i) \phantom{aaa}} \varphi(y^{\oplus i})\xrightarrow{\phantom{aa} \oplus \eta_{y^{\oplus i}}(j) \phantom{aaa}} \varphi(z^{\oplus i}) \xrightarrow{\phantom{aa} \oplus \eta_{z^{\oplus i}}(i) \phantom{aaa}} \varphi(z),
        \end{equation*}
    and as a consequence, we have ---
    \begin{equation}
                \label{eq:cycle}
                \left(\left(\left(\left(\varphi(z)\right)^{\oplus \eta_z(j)}\right)^{\oplus \eta_y(i)}\right)^{\oplus \eta_{y^{\oplus i}}(j)}\right)^{\oplus\eta_{z^{\oplus i}}(i)}=\varphi(z).
    \end{equation}
    Furthermore, by a combination of Observations~\ref{obs:different-direction-index} and~\ref{obs:reflection-invariance}, we have $\eta_z(j)\neq \eta_y(i)$, $\eta_y(i)\neq \eta_{y^{\oplus i}}(j)$, and $\eta_{y^{\oplus i}}(j)\neq \eta_{z^{\oplus i}}(i)$. Therefore, if $\eta_z(i)\neq \eta_y(i)$, and consequently by Observation~\ref{obs:reflection-invariance}, $\eta_{z^{\oplus i}}(i) \neq \eta_y(i)$, then Equation~\ref{eq:cycle} cannot hold since there would be no way to cancel out all the bit-flips. Thus, we conclude that $\eta_z(i)=\eta_y(i)$. 
    \end{proof}

    Lemma~\ref{lemma:same-direction-index} tells us that the functions $\eta_x:[n]\to [N]$ are identical for all $x\in\{0, 1\}^n$. Therefore, we can drop the subscript and refer to this function as $\eta$. By Observation~\ref{obs:different-direction-index}, $\eta$ is injective. One way of interpreting Lemma~\ref{lemma:same-direction-index} is that regardless of the input string, as long as the location of the bit flip in the input is fixed, the location of the corresponding bit flip in the output is also fixed. We now argue something stronger --- regardless of the input string, as long as the location and \textit{direction}\footnote{There are two possible directions for a bit flip --- flipping a 0 to a 1, and flipping a 1 to a 0.} of the bit flip in the input is fixed, the location and direction of the corresponding bit flip in the output is also fixed.
    \begin{lemma}
        \label{lemma:location-and-direction}
        For all $x, y\in \{0, 1\}^n$ and $i\in [n]$, $\varphi(x)[\eta(i)]=\varphi(y)[\eta(i)]$ if and only if $x[i]=y[i]$. 
    \end{lemma}
    \begin{proof}
    \renewcommand{\qedsymbol}{$\lrcorner$}
        Consider any shortest sequence of bit flips that transforms $x$ into $y$. Such a sequence also specifies a unique sequence of bit flips transforming $\varphi(x)$ into $\varphi(y)$. If $x[i]=y[i]$, the $i^{\text{th}}$ bit of $x$ is never touched during the sequence. Since $\eta$ is injective, the $\eta(i)^{\text{th}}$ bit of $\varphi(x)$ is also never touched and $\varphi(x)[\eta(i)]=\varphi(y)[\eta(i)]$.

        Similarly, if $x[i]\neq y[i]$, any shortest sequence of bit flips transforming $x$ into $y$ must flip the $i^{\text{th}}$ bit exactly once. Since $\eta$ is injective, in the corresponding sequence transforming $\varphi(x)$ into $\varphi(y)$, the $\eta(i)^{\text{th}}$ bit is flipped exactly once. Therefore, $\varphi(x)[\eta(i)]\neq\varphi(y)[\eta(i)]$.  
    \end{proof}

    Now, for each $i\in[n]$, we define the function $\pi_i:\{0, 1\}\to \{0, 1\}$ as follows. Pick any arbitrary $x\in \{0, 1\}^n$. If $\varphi(x)[i]=x[i]$, we define $\pi_i(b):=b$ for all $b\in \{0, 1\}$. Otherwise, we define $\pi_i(b):= 1-b$ for all $b\in \{0, 1\}$. By Lemma~\ref{lemma:location-and-direction}, $\varphi(x)[\eta(i)]=\pi_i(x[i])$ for all $x\in \{0, 1\}^n$ and $i\in [n]$. The final piece of the proof is the following lemma.

    \begin{lemma}
        Let $S\subseteq[N]$ be the set of indices that are not in the image of $\eta$, i.e., $S=\{j\in [N]: \nexists i \textnormal{ with }\eta(i) =j\}$. Then for all $x, y\in \{0, 1\}^n$, we have $\varphi(x)_{|S}=\varphi(y)_{|S}$.
    \end{lemma}
    \begin{proof}
    \renewcommand{\qedsymbol}{$\lrcorner$}
        Assume for the sake of contradiction that there exist $x, y\in \{0, 1\}^n$ and $j\in S$ such that $\varphi(x)[j]\neq \varphi(y)[j]$. Consider any shortest sequence of bit flips transforming $x$ into $y$. Such a sequence also specifies a unique sequence of bit-flips transforming $\varphi(x)$ into $\varphi(y)$. This sequence only affects indices that are in the image of $\eta$. Since $j$ is not in the image of $\eta$, the $j^{\text{th}}$ bit of $\varphi(x)$ is never touched during this process. This means $\varphi(x)[j]=\varphi(y)[j]$, which is a contradiction.
    \end{proof}

    With this, we have now shown that $\varphi$ satisfies all the properties of an interleaved embedding, and we are done.
\end{proof}

\subsection{Extending to Larger Alphabets}

\label{sec:isometry-implies-generalized-interleaved-embedding}

Even though Theorem~\ref{thm:isometry-implies-interleaving} is stated for embeddings of strings over the binary alphabet, a similar result holds true for larger alphabets as well. Let $\Sigma$ be any alphabet and consider any isometric embedding $\varphi:\Sigma^n\to \Sigma^N$ of the Hamming metric into the edit metric. We claim that even in this more general setting, $\varphi$ must, in some sense, be an interleaved embedding.

To show this, we first need to define what it means for an embedding to be interleaved when working with a non-binary alphabet, since Definition~\ref{def:interleaved-embedding} is specifically tailored to the binary case. However, finding the right generalization is not too difficult here. A closer look at Definition~\ref{def:interleaved-embedding} reveals that its dependence on the binary alphabet stems from the functions $\pi_i:\{0, 1\}\to \{0, 1\}$, where each $\pi_i$ either complements its input bit or leaves it unchanged. In the general setting, we simply require these functions to be \textit{permutations} of the alphabet, meaning that each $\pi_i:\Sigma\to\Sigma$ is a bijection that maps the alphabet onto itself.

\begin{defi}[Generalized Interleaved Embedding]
\label{def:generalized-interleaved-embedding}
    Let $\Sigma$ be an alphabet. A function $\varphi:\Sigma^n \to \Sigma^N$ with $N\geq n$ is called a {generalized interleaved embedding} if there exist ---
    \begin{itemize}
        \item an injective function $\eta : [n] \to [N]$,
        \item a collection of bijective functions $\pi_1, \pi_2, \ldots , \pi_n: \Sigma \to \Sigma$, and
        \item a fixed string $w\in \Sigma^{N-n}$,
    \end{itemize}
    such that for all $x\in\Sigma^n$, if $X=\varphi(x)$, then the following hold.
    \begin{itemize}
        \item For all $i\in [n]$, $X[\eta(i)]= \pi_i(x[i])$.
        \item Let $S\subseteq [N]$ be the set of indices that are not in the image of $\eta$, i.e., $S=\{j\in [N]: \nexists i \textnormal{ with } \eta(i) = j\}$. Then $X_{|S}=w$.
    \end{itemize}
\end{defi}

Note that Definition~\ref{def:generalized-interleaved-embedding} reduces to Definition~\ref{def:interleaved-embedding} in the case when $\Sigma=\{0, 1\}$. Therefore, from now on, when we refer to interleaved embeddings, we will mean functions defined in Definition~\ref{def:generalized-interleaved-embedding}.\footnote{In particular, we drop the word ``generalized''.} We are now ready to state the main result of this section in full generality.

\begin{theorem}
    \label{thm:isometry-implies-generalized-interleaving}
    Let $\Sigma$ be an alphabet and $\varphi:\Sigma^n\to \Sigma^N$ be an isometric embedding of the Hamming metric into the edit metric. Then $\varphi$ is an interleaved embedding.
\end{theorem}

The proof of Theorem~\ref{thm:isometry-implies-generalized-interleaving} is essentially the same as that of Theorem~\ref{thm:isometry-implies-interleaving}. The only difference is that we cannot argue using bit flips anymore. However, this is not an issue. We now argue that as long as a single index is modified, regardless of \textit{how} it is changed, it always affects the same index at the output. 

More precisely, fix some $x\in \Sigma^n$, an index $i\in[n]$ and symbols $\sigma, \sigma'\in \Sigma$ different from $x[i]$. Let $y$ be the string obtained from $x$ by replacing the $i^{\text{th}}$ symbol by $\sigma$, and let $y'$ be the string obtained from $x$ by replacing the $i^{\text{th}}$ symbol by $\sigma'$. By a similar argument as in the binary case, there exists some $k\in [N]$ such that $\varphi(y)$ is obtained by changing the $k^{\text{th}}$ symbol of $\varphi(x)$. Similarly, there must also exist $k'\in[N]$ such that $\varphi(y')$ is obtained by changing the $k'^{\text{th}}$ symbol of $\varphi(x)$. Since $\ham(y, y')\leq 1$, if $k\neq k'$, then $\ed(\varphi(y), \varphi(y'))=2$, contradicting that $\varphi$ is isometric. Thus, we must have $k=k'$.

What this means is that as long as all other symbols remain fixed, any sequence of modifications to the $i^{\text{th}}$ symbol always affects the same index in $\varphi(x)$. The remainder of the proof would then proceed exactly as in Theorem~\ref{thm:isometry-implies-interleaving}.

\section{Upper Bounds on the Rate}
\label{sec:UpperBoundsOnTheRate}

In this section, we show that for every alphabet $\Sigma$, any isometric embedding $\varphi: \Sigma^n \to \Sigma^N$ of the Hamming metric into the edit metric must have rate bounded away from $\frac{1}{2}$, for sufficiently large $n$.

\begin{theorem}
    \label{thm:rate-upper-bound}
    For every alphabet $\Sigma$, there exists an integer $n_0$ such that every isometric embedding $\varphi: \Sigma^n \to \Sigma^N$ of the Hamming metric into the edit metric with $n\geq n_0$ must have rate at most $\frac{1}{2}-\frac{1}{16|\Sigma|}$. 
\end{theorem}
\begin{proof}
    Assume, for the sake of contradiction, that there exists an alphabet $\Sigma$ and an isometric embedding $\varphi:\Sigma^n \to \Sigma^N$ of the Hamming metric into the edit metric with rate equal to $\frac{1}{2}-\frac{1}{16|\Sigma|}+\rho$ for some constant $\rho >0$. We will show that this leads to a contradiction if $n$ is sufficiently large. In particular, we will show that for sufficiently large $n$, there exists a pair of strings $x^*, y^*\in \Sigma^n$ such that $\ed(\varphi(x^*), \varphi(y^*)) < \ham(x^*, y^*)$.

    In order to find such $x^*$ and $y^*$,  we will employ the probabilistic method. First, we will define a collection $\mathcal{C}$ of specially chosen triples $(x, y, \mathcal{A})$, where $x, y\in \Sigma^n$ and $\mathcal{A}$ is an edit distance alignment between $\varphi(x)$ and $\varphi(y)$. Then, we will pick a triple $(x, y, \mathcal{A})\in \mathcal{C}$ uniformly at random and show that the expected cost of $\mathcal{A}$ with respect to $\varphi(x)$ and $\varphi(y)$ is smaller than the minimum possible Hamming distance between $x$ and $y$. This will guarantee the existence of the desired $x^*$ and $y^*$. Details follow.
    
    Let $\Delta$ be a positive divisor of $N$ to be chosen later. The exact value of $\Delta$ will depend on $\rho$ but not $N$. Define $S_{\textnormal{shifts}} := \{-(\Delta-1), -(\Delta-2), \ldots , -1\}\cup \{1, 2, \ldots , (\Delta-1)\}$, and for any $\delta\in S_{\textnormal{shifts}}$, define $\mathcal{A}_\delta$ to be the following edit distance alignment between two length $N$ strings.
    \[\mathcal{A}_\delta := \{(i, i+\delta): \max\{1, 1-\delta\} \leq i\leq \min\{N-\delta, N\}\}\]
    In words, for any pair of length $N$ strings $X$ and $Y$, $\mathcal{A}_\delta$ aligns every symbol in $X$ to the $\delta^{\textnormal{th}}$ next symbol in $Y$. Now with each alignment $\mathcal{A}_\delta$ where $\delta \in S_{\textnormal{shifts}}$, we will associate a pair of strings $x_\delta, y_\delta \in \Sigma^n$. But before we describe $x_\delta$ and $y_\delta$, let us first recall some properties of interleaved embeddings. Since $\varphi$ is isometric, by Theorem~\ref{thm:isometry-implies-generalized-interleaving}, $\varphi$ is an interleaved embedding. Consequently, there exists $I_{\textnormal{frozen}}\subseteq [N]$ such that $\varphi(x)_{|I_{\textnormal{frozen}}}$ is the same string for all $x\in \Sigma^n$. Recalling Definition~\ref{def:generalized-interleaved-embedding}, $I_{\textnormal{frozen}}$ is simply the set of indices not in the image of $\eta$, the injective function associated with $\varphi$. Define $I_{\textnormal{mutable}} := [N]\setminus I_{\textnormal{frozen}}$. Note that for any $w\in \Sigma^{|I_{\textnormal{mutable}}|}$, we can always choose $x\in \Sigma^n$ such that $\varphi(x)_{|I_{\textnormal{mutable}}}$ is set to $w$. Furthermore, we can set each symbol independently since for all $i\in I_{\textnormal{mutable}}$ and $\sigma\in \Sigma$, there exist unique $j\in [n]$ and $\sigma'\in \Sigma$ such that $\varphi(x)[i] = \sigma$ if and only if $x[j]=\sigma'$. For every $x\in \Sigma^n$ and $i\in [N]$, let us call the symbol $\varphi(x)[i]$ \textit{mutable} if $i\in I_{\textnormal{mutable}}$ and \textit{frozen} otherwise.
    
    Given $\delta \in S_{\textnormal{shifts}}$, we choose a pair of strings $x_\delta, y_\delta\in \Sigma^n$ satisfying the following properties.
    \begin{itemize}
        \item $\ham(x_\delta, y_\delta) = n$,
        \item For each $i\in I_{\textnormal{mutable}}$ such that $i+\delta\in [N]$, we have $\varphi(x_\delta)[i] = \varphi(y_\delta)[i+\delta]$.
    \end{itemize}
    Essentially, what we want from $x_\delta$ and $y_\delta$ is that whenever some mutable symbol of $\varphi(x_\delta)$ is aligned with some (not necessarily mutable) symbol of $\varphi(y_\delta)$, then it should result in a match and not a substitution. We additionally want $x_\delta$ and $y_\delta$ to have the maximum possible Hamming distance. Note that such $x_\delta$ and $y_\delta$ can always be found. Specifically, for $\delta >0$, we can set $\varphi(x_\delta)[i]$ and $\varphi(y_\delta)[i]$ for $i\in I_{\textnormal{mutable}}$ in descending order of $i$.\footnote{Although we say we are setting \(\varphi(x_{\delta})[i]\) and \(\varphi(y_{\delta})[i]\), what we are actually setting are \(x_{\delta}[\eta^{-1}(i)]\) and \(y_{\delta}[\eta^{-1}(i)]\).} By doing so, every symbol in $\varphi(x_\delta)$ is aligned with some symbol (whether mutable or not) in $\varphi(y_\delta)$ that has already been set. During each iteration, we begin by setting $\varphi(x_\delta)[i]$ to match the symbol it is aligned with. Then we assign $\varphi(y_\delta)[i]$ a symbol different from $\varphi(x_\delta)[i]$, and move on to the next (smaller) value of $i$. This ensures that $x_\delta$ and $y_\delta$ meet the desired requirements. For $\delta<0$, we can follow a similar procedure but set $\varphi(x_\delta)[i]$ and $\varphi(y_\delta[i])$ in ascending order of $i$ instead. See Figure~\ref{fig:choosing-adversarial-strings} for an illustration of this process.

    \begin{figure}
        \centering
        \includegraphics[width=\linewidth]{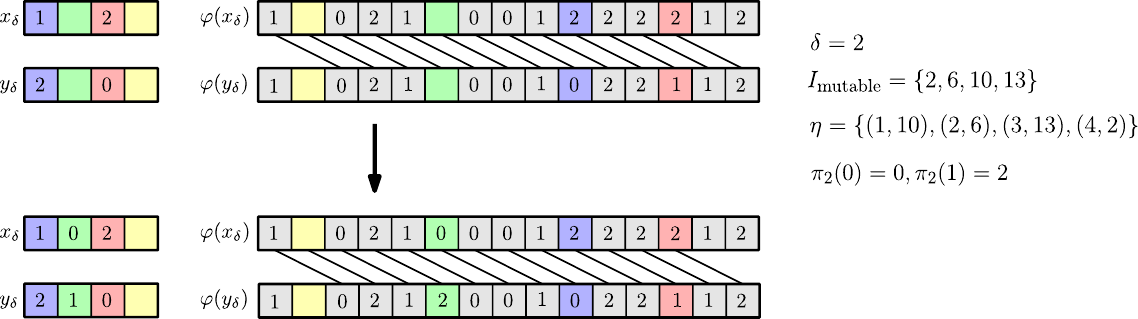}
        \caption{\(\eta\) is indicated by the matching colors. The red and purple symbols were set in previous iterations, and we now seek to set the green symbols. Specifically, we require the green symbol of \(\varphi(x_\delta)\) to be 0, while the green symbol of \(\varphi(y_\delta)\) must differ from that of \(\varphi(x_\delta)\). This can be achieved by appropriately setting the green symbols in \(x_\delta\) and \(y_\delta\).}
        \label{fig:choosing-adversarial-strings}
    \end{figure}

    We are now ready to describe the collection $\mathcal{C}$ alluded to earlier --- it consists of all triples of the form $(x_\delta, y_\delta, \mathcal{A}_\delta)$, where $\delta\in S_{\textnormal{shifts}}$. Let $\Tilde{\delta}$ be a uniformly random element of $S_{\textnormal{shifts}}$. The main step in the proof involves bounding the expected cost of $\mathcal{A}_{\Tilde{\delta}}$ with respect to $\varphi(x_{\Tilde{\delta}})$ and $\varphi(y_{\Tilde{\delta}})$.
    \begin{lemma}
        \label{lemma:expected-cost-bound}
        $\E\left[\cost_{\varphi(x_{\Tilde{\delta}})}^{\varphi(y_{\Tilde{\delta}})}\left(\mathcal{A}_{\Tilde{\delta}}\right)\right] < n$, for sufficiently large values of $n$.
    \end{lemma}
    \begin{proof}
        \renewcommand{\qedsymbol}{$\lrcorner$}
        For ease of notation, define $X := \varphi(x_{\Tilde{\delta}})$, $Y := \varphi(y_{\Tilde{\delta}})$, and $\mathcal{A} := \mathcal{A}_{\Tilde{\delta}}$. Furthermore, define $\varepsilon := \frac{1}{16|\Sigma|} -\rho$ so that the rate of $\varphi$ is $\frac{1}{2}-\varepsilon$. We will bound $\E[\cost_{X}^Y(\mathcal{A})]$ by individually bounding the expected number of insertions, deletions, and substitutions. Bounding the first two is easy --- no matter what $\Tilde{\delta}$ is, the numbers of insertions or deletions are both always at most $(\Delta-1)$. For substitutions, note that for any $i\in I_{\textnormal{mutable}}$, $X[i]$ is by definition never substituted by $\mathcal{A}$ --- it is either matched or deleted. Therefore, in order to bound the expected number of substitutions, it suffices to only focus on the frozen symbols. Let us define an index $i\in I_{\textnormal{frozen}}$ to be \textit{good} if there exists $j\in I_{\textnormal{frozen}}$ such that $(i, j)\in \mathcal{A}$ and $X[i]=Y[j]$. Less formally, a good index is where a \textit{match} occurs between a pair of frozen symbols. Let $I_{\textnormal{good}}\subseteq I_{\textnormal{frozen}}$ be the set of all good indices. Clearly, the number of substitutions is upper bounded by $|I_{\textnormal{frozen}}|-|I_{\textnormal{good}}|$. Since $\varphi$ has rate $\frac{1}{2}-\varepsilon$, $|I_{\textnormal{frozen}}| = N\left(\frac{1}{2}+\varepsilon\right)$, and we have ---
        \[\E\left[\cost_X^Y(\mathcal{A})\right] \leq 2(\Delta -1) + N\left(\frac{1}{2}+\varepsilon\right)-\E\left[|I_{\textnormal{good}}|\right]\]
        Let $\mathbb{1}_{I_{\textnormal{good}}}:I_\textnormal{frozen}\to \{0, 1\}$ be the indicator function of $I_{\textnormal{good}}\subseteq I_\textnormal{frozen}$. By linearity of expectation, we have ---
        \[\E\left[|I_{\textnormal{good}}|\right] = \sum_{i\in I_\textnormal{frozen}}\mathbb{1}_{I_{\textnormal{good}}}(i)=\sum_{i\in I_\textnormal{frozen}}\Pr[i\in I_{\textnormal{good}}]\]
    So, we need to obtain a lower bound on the probability that an index is good. To do so, we partition $X$ and $Y$ into contiguous length-$\Delta$ substrings. For each $t\in [\frac{N}{\Delta}]$ and $\sigma\in \Sigma$, we define $c_{t, \sigma}$ to be the number of times symbol $\sigma$ appears as a frozen symbol in the $t^{\textnormal{th}}$ substring in this partition, i.e., 
    \[c_{t, \sigma} := \left|\{i\in [t(\Delta-1), t\Delta] : i\in I_\textnormal{frozen} \textnormal{ and } X[i]=Y[i]=\sigma\}\right|\]
    Next, we make the following observation.
    \begin{obs}
        Let $i\in I_\textnormal{frozen}$ and let $t\in[\frac{N}{\Delta}]$ be such that $i\in [t(\Delta-1), t\Delta]$. Then we have ---
        \[\Pr[i\in I_{\textnormal{good}}] \geq \frac{1}{2}\cdot\frac{c_{t, X[i]}-1}{\Delta -1}\]
    \end{obs}
    \begin{proof}
        \renewcommand{\qedsymbol}{$\lrcorner$}
        We have ---
        \allowdisplaybreaks
        \begin{small}
        \begin{align*}
            \Pr[i\in I_{\textnormal{good}}] &\geq \Pr\left[\left(i+\Tilde{\delta}\in [t(\Delta-1), t\Delta] \right)\wedge \left(\left(i+\Tilde{\delta} \in I_\textnormal{frozen}\right) \wedge \left(X[i]=Y[i+\Tilde{\delta}]\right)\right)\right]\\
            &= \Pr\left[i+\Tilde{\delta}\in [t(\Delta-1), t\Delta]\right]\cdot \Pr\left[\left(i+\Tilde{\delta} \in I_\textnormal{frozen}\right) \wedge \left(X[i]=Y[i+\Tilde{\delta}]\right)\middle| i+\Tilde{\delta}\in [t(\Delta-1), t\Delta]\right]\\
            &= \frac{1}{2}\cdot\frac{c_{t, X[i]}-1}{\Delta -1}        \qedhere\\
        \end{align*}
                \end{small}
    \end{proof}
So, we can write ---
    \begin{equation}
    \label{ineq:rhs}
        \E\left[|I_{\textnormal{good}}|\right]\geq \frac{1}{2(\Delta-1)}\sum_{{t\in \left[\frac{N}{\Delta}\right]}}\sum_{{{\sigma\in \Sigma}}} c_{t, \sigma}(c_{t, \sigma}-1)
    \end{equation}
    The right hand side of~(\ref{ineq:rhs}) is minimized when $c_{t, \sigma}$'s are all the same and equal to $\frac{N(\frac{1}{2}+\varepsilon)}{\frac{N|\Sigma|}{\Delta}}=\frac{\Delta(\frac{1}{2}+\varepsilon)}{|\Sigma|}$. Furthermore, each $c_{t, \sigma}$ can be replaced with $\frac{\Delta}{2|\Sigma|}$ to obtain a value that is strictly smaller than the minimum value of the right hand side of (\ref{ineq:rhs}). Thus, we have ---
    \allowdisplaybreaks
    \begin{align*}
        \E\left[|I_{\textnormal{good}}|\right]&\geq \frac{1}{2(\Delta-1)}\sum_{{t\in \left[\frac{N}{\Delta}\right]}}\sum_{{{\sigma\in \Sigma}}} c_{t, \sigma}(c_{t, \sigma}-1)\\
        & > \frac{1}{2(\Delta -1)}\cdot\frac{N}{\Delta}\cdot|\Sigma|\cdot \frac{\Delta}{2|\Sigma|}\left(\frac{\Delta}{2|\Sigma|}-1\right)\\
        &\geq \frac{1}{2\Delta}\cdot \frac{N}{2}\left(\frac{\Delta}{2|\Sigma|}-1\right)\\
        & = N\left(\frac{1}{8|\Sigma|}-\frac{1}{4\Delta}\right),
    \end{align*}
    and thus, we have ---
    \[\E\left[\cost_X^Y(\mathcal{A})\right] < 2(\Delta -1) + N\left(\frac{1}{2}+\varepsilon\right)-N\left(\frac{1}{8|\Sigma|}-\frac{1}{4\Delta}\right)\]. Consequently, we have ---
    \begin{align*}\allowdisplaybreaks
        n-\E\left[\cost_X^Y(\mathcal{A})\right] & > n - 2(\Delta -1) - N\left(\frac{1}{2}+\varepsilon\right)+N\left(\frac{1}{8|\Sigma|}-\frac{1}{4\Delta}\right)\\
        &= N\left(\frac{1}{2}-\varepsilon\right) - 2(\Delta -1) - N\left(\frac{1}{2}+\varepsilon\right)+N\left(\frac{1}{8|\Sigma|}-\frac{1}{4\Delta}\right)\\
        &= N\left(-2\varepsilon + \frac{1}{8|\Sigma|}-\frac{1}{4\Delta}\right) - 2(\Delta -1)
    \end{align*}
    Now recall that we had $\varepsilon=\frac{1}{16|\Sigma|}-\rho$. Thus, if we choose $\Delta$ to be large enough such that $\frac{1}{8\Delta}<\rho$, we will have $\varepsilon < \frac{1}{16|\Sigma|}-\frac{1}{8\Delta}$ and,
    \begin{equation}
    \label{eq:cost-diff}
        n - \E\left[\cost_X^Y(\mathcal{A})\right] > \alpha N-2(\Delta-1),
    \end{equation}
    where $\alpha := -2\varepsilon + \frac{1}{8|\Sigma|}-\frac{1}{4\Delta} > 0$. Finally, we note that by making $n$ and, consequently, $N$ sufficiently large, one can make the right-hand side of (\ref{eq:cost-diff}) positive. The conclusion then follows.
    \end{proof}

    Lemma~\ref{lemma:expected-cost-bound} guarantees the existence of a triple $(x_\delta, y_\delta, \mathcal{A}_\delta)\in \mathcal{C}$ such that $\cost_{\varphi(x_\delta)}^{\varphi(y_\delta)}(\mathcal{A}_\delta) < \ham(x_\delta, y_\delta)$. However, this contradicts the assumption that $\varphi$ is an isometric embedding of the Hamming metric into the edit metric.
\end{proof}

We can further use Theorem~\ref{thm:rate-upper-bound} to immediately get an upper bound on the rate on any \(1\)-embedding.

\begin{theorem}
\label{thm:general-rate-upper-bound-non-isometric}
For every alphabet $\Sigma$, there exists an integer $n_0$ such that every $1$-embedding $\varphi: \Sigma^n \to \Sigma^N$ of the Hamming metric into the edit metric with $n\geq n_0$ must have rate at most $\frac{1}{2}$.
\end{theorem}

\begin{proof}

\begin{sloppypar}
If $\varphi$ is an isometric embedding, then Theorem~\ref{thm:rate-upper-bound} implies a rate strictly less than $\frac{1}{2}$. It remains to prove the statement for the case where $\varphi$ is a non-isometric $1$-embedding with $\ed~(\varphi(x),~\varphi(y))=K~\cdot~\ham(x, y) $, for all $x,  y \in \{0, 1\}^n$ for some $K$ not equal to $1$. For two strings $x',y'$ with Hamming distance $1$ we have that $\ed~(\varphi(x'),~\varphi(y'))=K$ and therefore $K$ must be a positive integer and therefore at least two. The rate $n/N$ of $\varphi$ is furthermore at most $1/K$ because for any two strings $x'',y''$ of maximal Hamming distance $n$ it holds that $N \geq \ed~(\varphi(x''),~\varphi(y'')) = K \cdot n$. This shows that the rate of $\varphi$ is at most $\frac{1}{2}$.\qedhere
\end{sloppypar}
\end{proof}

Finally, we briefly discuss about 1-embedding with rate higher than $1/3$. Consider any family of 1-embeddings $\{\varphi_n: \{0,1\}^n \to \{0,1\}^N\}_{n\in\mathbb{N}}$ of the Hamming metric into the edit metric such that for all $n\in \mathbb{N}$, we have $N<3n$. 
If $\varphi_n$ is  isometric, then by Theorem~\ref{thm:isometry-implies-interleaving}, it must admit an interleaved structure as given in Definition~\ref{def:interleaved-embedding}.

On the other hand, if $\varphi_n$ is not isometric, then by the same argument as in the proof of Theorem~\ref{thm:general-rate-upper-bound-non-isometric}, we have that for all $x,y\in \{0,1\}^n$, it must be that $\ed(\varphi_n(x),\varphi_n(y))=2\cdot \ham(x,y)$. We conjecture that even in this case, the embedding must have a structure similar to Definition~\ref{def:interleaved-embedding}. 

\begin{ques}\label{open:structure}
    Fix some $n,N\in \mathbb{N}$ such that $N<3n$. Let $\varphi:\{0,1\}^n \rightarrow \{0,1\}^N$ be a non-isometric 1-embedding of the Hamming metric into the edit metric. Then, prove that there exist disjoint sets of indices $G_1, \dots, G_n,S \subset [N]$, $|G_i|=2$, strings \(W_1, W_2, \ldots , W_n \in \{0, 1\}^2\), $W_S\in \{0,1\}^{N-2n}$, such that $\varphi(x)_{\lvert G_i}=W_i$ if $x_i=0$ and $\varphi(x)_{\lvert G_i}=\overline{W_i}$ if $x_i=1$. Furthermore, the symbols at $S$ are $W_S$ regardless of $x$.
\end{ques}

We remark below that if we have a positive resolution of the above open question, then we can conclude that every 1-embedding with rate more than $1/3$ must be isometric. This would concretely justify the interest to only search for high rate isometric embeddings.

\begin{remark}\label{rem:1embed}
Assuming a positive resolution to Open Question~\ref{open:structure}, we now argue that any family of $1$-embeddings with rate greater than $1/3$ must be isometric. Suppose not, and fix a non-isometric $1$-embedding $\varphi:\{0,1\}^n \to \{0,1\}^N$ with rate greater than $1/3$. Then $\varphi$ admits an interleaved structure as in Open Question~\ref{open:structure}. As in the proof of Theorem~\ref{thm:rate-upper-bound}, we can now set up an appropriate collection of triples $(x,y,\mathcal{A})$ and apply the probabilistic method to derive a contradiction. The only modification to the argument is that we can no longer assert that the ``mutable'' symbols are never substituted. Nevertheless, we can choose strings in a way such that at most half of the ``mutable'' symbols are substituted, and this suffices, since we only need to rule out rates exceeding $1/3$.
\end{remark}

\section{Breaking the Rate Barrier via Larger Output Alphabets}

\label{sec:close-to-1}

Up to this point, we have only considered embeddings where the input and output alphabets have the same size. Theorem~\ref{thm:rate-upper-bound} tells us that in this setting, the rate is necessarily bounded away from $\frac{1}{2}$. While this is the most natural setting, one can also consider embeddings with \textit{different-sized} input and output alphabets. It turns out that by expanding the output alphabet, it is possible to find isometric embeddings with rates \textit{arbitrarily close to 1}! However, before explaining how this is done, we must first refine the definition of rate for such embeddings. In the case where the input and output alphabets are of equal size, the rate is defined as the ratio of the input string’s length to the output string’s length. When using a larger output alphabet, the appropriate measure to consider instead is not length but the \textit{number of bits}. 

\begin{defi}
\label{def:rate}
    The rate of any map $\varphi:\Sigma_{\textnormal{in}}\to \Sigma_{\textnormal{out}}$ is the ratio $\frac{n\log|\sigmain|}{N\log|\sigmaout|}$.
\end{defi}
    
Recall the naive embedding from Section~\ref{sec:sync-proof-overview}, which inserts a new symbol after each input bit. For that embedding, we have $N=2n$, $\Sigma_{\text{in}} = \{0, 1\}$, and $\Sigma_{\text{out}} = \{0, 1, \dots, n+1\}$. Thus, even though $\frac{n}{N}$ is as large as $\frac{1}{2}$, the rate, as defined in Definition~\ref{def:rate}, is actually only $\Theta(\frac{1}{\log n})$.

We are now ready to state the main result of this section.
\begin{theorem}
\label{thm:rate-close-to-1}
    For every $\rho>0$, there exist alphabets $\sigmain$, $\sigmaout$ and a family of functions $\varphi_n:\sigmain^n\to \sigmaout^{N}$ for each positive integer $n$ such that the following hold.
    \begin{enumerate}
        \item For every $n$, $\varphi_n$ is an isometric embedding of the Hamming metric into the edit metric, i.e., for all $x, y\in \sigmain^n$, we have $\ed(\varphi_n(x), \varphi_n(y))=\ham(x, y)$.
        \item The rate of $\varphi_n$ is at least $(1-\rho)$.
    \end{enumerate}
\end{theorem}

\noindent Before we prove Theorem~\ref{thm:rate-close-to-1}, we first make the following observation about locally self-matching strings.

\begin{obs}
\label{obs:lcs-implies-edit}
    Let $w$ be an $\varepsilon$-locally self matching string. Then for every substring $s$ of $w$, we have:
    \[\selfed(s, s) \geq (1-\varepsilon)|s|\]
\end{obs}
\begin{proof}
    Let $\mathcal{A}$ be an optimal nowhere-vertical edit distance alignment between $s$ and itself. We have:
    \[\selfed(s, s) =\cost_s^s(\mathcal{A})\geq|S_s^s(\mathcal{A})|+|D_s^s(\mathcal{A})|\geq |s|-\slcs(s, s)\geq (1-\varepsilon)|s|\qedhere\]
\end{proof}

\begin{proof}[Proof of Theorem~\ref{thm:rate-close-to-1}]
Fix any $\rho>0$. 
Define $m:=\left\lceil\frac{C(1-\rho)+1}{\rho}\right\rceil$, where $C\geq 1$ is a constant to be chosen later. Furthermore, let $\varepsilon :=\frac{1}{m}$ and $\varepsilon' = \frac{\varepsilon}{2}$. We will choose $\sigmain=\left[2^{\frac{1}{\varepsilon}\log(\frac{1}{\varepsilon})}\right]$ as our input alphabet. For our output alphabet, we choose $\sigmaout = \sigmain\times \Sigma'$, where $\Sigma':=\left[\frac{5e^2}{\varepsilon'^2}\right]$.

Fix any positive integer $n$ and define $N:= n+\left\lfloor\frac{n}{m-1}\right\rfloor$. Our embedding map $\varphi_n:\sigmain^n\to \sigmaout^N$ is defined as follows. First, we fix some $\varepsilon'$-locally self-matching string $w$ of length $N$ over the alphabet $\Sigma'$. Since $\varepsilon' \leq \frac{1}{2}$, by Theorem~\ref{thm:lll-string}, such a string $w$ always exists. We also fix some arbitrary symbol $\sigma\in \sigmain$. Next, given any $x\in \sigmain^n$, we construct $x_{\textnormal{padded}}\in \sigmain^N$ by inserting the symbol $\sigma$ after every $(m-1)$ symbols in $x$. Finally, we define $\varphi_n(x)\in \sigmaout^N$ to be the \textit{product} of the strings $x_{\textnormal{padded}}$ and $w$. In other words, we define:
\[\varphi_n(x):= (x_{\textnormal{padded}}[1], w[1])\circ (x_{\textnormal{padded}}[2], w[2])\circ \cdots \circ (x_{\textnormal{padded}}[N],w[N])\]
We first claim that if $C$ is appropriately chosen, then $\varphi_n$ indeed has rate at least $1-\rho$.
\begin{claim}
    If $C\geq 10$, then the rate of $\varphi_n$ is at least $1-\rho$.
\end{claim}
\begin{proof}
    \renewcommand{\qedsymbol}{$\lrcorner$}
    First note that $N=n+\left\lfloor\frac{n}{m-1}\right\rfloor \leq \frac{nm}{m-1}=\frac{n}{1-\varepsilon}$. Therefore, $\frac{n}{N}\geq 1-\varepsilon$. Furthermore, since $C\geq 1$ and $\rho<1$, $m=\left\lceil\frac{C(1-\rho)+1)}{\rho}\right\rceil \geq 2$. Consequently, $\varepsilon\leq \frac{1}{2}$ and we have $\log|\Sigma'|= \log\left(\frac{5e^2}{\varepsilon'^2}\right)=\log(20e^2)+2\log\left(\frac{1}{\varepsilon}\right) \leq C\log\left(\frac{1}{\varepsilon}\right)$ for all $\varepsilon\leq \frac{1}{2}$ if $C\geq 10$. Therefore, the rate of $\varphi_n$ is:
    \allowdisplaybreaks
    \begin{align*}
        \frac{n\log|\sigmain|}{N\log|\sigmaout|}&=\frac{n}{N}\cdot\frac{\log|\sigmain|}{\log|\sigmain|+\log|\Sigma'|}\\
        &\geq(1-\varepsilon)\cdot\frac{\log\left(2^{\frac{1}{\varepsilon}\log\left(\frac{1}{\varepsilon}\right)}\right)}{\log\left(2^{\frac{1}{\varepsilon}\log\left(\frac{1}{\varepsilon}\right)}\right)+C\log\left(\frac{1}{\varepsilon}\right)}\\
        &=(1-\varepsilon)\cdot \frac{\frac{1}{\varepsilon}\log\left(\frac{1}{\varepsilon}\right)}{\frac{1}{\varepsilon}\log\left(\frac{1}{\varepsilon}\right)+C\log\left(\frac{1}{\varepsilon}\right)}\\
        &=\frac{1-\varepsilon}{1+C\varepsilon}\\
        &=1-\frac{(C+1)\varepsilon}{1+C\varepsilon}\\
    \end{align*}
    Now recall that $\varepsilon\leq \frac{\rho}{C(1-\rho)+1}$, and so, $\rho \geq \frac{(C+1)\varepsilon}{1+C\varepsilon}$. The conclusion then follows.
\end{proof}
Next, we show that $\varphi_n$ is indeed isometric for every positive integer $n$. Suppose not; then there exists a positive integer $n$ and $x, y\in \sigmain^n$ such that $\ed(\varphi_n(x), \varphi_n(y))<\ham(x, y)$. Let $X:= \varphi_n(x)$, $Y:=\varphi_n(y)$, and $\mathcal{A}$ be any optimal edit distance alignment between $X$ and $Y$. Since $\ham(x, y)=\ham(X, Y)$, by assumption, we have $\cost(\mathcal{A}) <\ham(X, Y)$. Similar to the proof of Theorem~\ref{thm:main}, consider the partition of $[N]$ into alternating maximal vertical and maximal nowhere-vertical intervals under $\mathcal{A}$. By the same argument as in the proof of Theorem~\ref{thm:main}, there must exist a nowhere vertical interval $I$ in this partition such that $\cost(\mathcal{A}_{|I})<\ham(X_{|I}, Y_{|I})$. We first note 
 that $|I|\geq \frac{2}{\varepsilon}$. Suppose not. Then, since $w$ is a $\frac{\varepsilon}{2}$-locally self-matching string, our construction forces $w_{|I}$ to have all distinct symbols. In this case, we have:
 \[\cost(\mathcal{A}_{|I})\geq \selfed(X_{|I}, Y_{|I}) \geq \selfed(w_{|I}, w_{|I}) = |I|+2 > \ham(X_{|I}, Y_{|I}).\]
 Therefore, we may assume that $|I|\geq \frac{2}{\varepsilon}$. But for any $|I|\geq \frac{2}{\varepsilon}$, we have $(1-\frac{\varepsilon}{2})|I| \geq (1-\varepsilon)|I|+1$. And since in $X$ and $Y$, every $(\frac{1}{\varepsilon})^{\text{th}}$ symbol is the same, $(1-\varepsilon)|I|+1 $ is an upper bound on $\ham(X_{|I}, Y_{|I})$.  Combining all this, we have:
 \[\cost(\mathcal{A}_{|I})\geq \selfed(X_{|I}, Y_{|I}) \geq \selfed(w_{|I}, w_{|I}) \geq \left(1-\frac{\varepsilon}{2}\right)|I| \geq  (1-\varepsilon)|I|+1 \geq \ham(X_{|I}, Y_{|I}),\]
 which is a contradiction.
\end{proof}

\subsection*{Acknowledgements}
Part of the work of Sudatta Bhattacharya, Sanjana Dey, Elazar Goldenberg and Michal Kouck\'y
was carried out during a visit to DIMACS, with support from the National Science Foundation
under grant number CCF-1836666. Sudatta Bhattacharya and Michal Koucky were partially supported by the Grant Agency of the Czech Republic under the grant agreement no. 24-10306S and
by Charles Univ. project UNCE 24/SCI/008. Mursalin Habib, Bernhard Haeupler, and Karthik C. S. were partially funded by the Ministry of Education and Science of Bulgaria’s support for
INSAIT as part of the Bulgarian National Roadmap for Research Infrastructure. Bernhard Haeupler was partially funded through the European Research Council (ERC) under the European
Union’s Horizon 2020 research and innovation program (ERC grant agreement 949272). Elazar
Goldenberg, Mursalin Habib, and Karthik C. S. were partially supported by the National Science
Foundation Grant CCF-2313372. Mursalin Habib and Karthik C. S. were partially supported by
the National Science Foundation Grant CCF-2443697 and a grant from the Simons Foundation,
Grant Number 825876, Awardee Thu D. Nguyen. Sanjana Dey was partially supported by the
Fonds de la Recherche Scientifique – FNRS under Grant n° T.0188.23 (PDR ControlleRS).

\bibliographystyle{alpha}
\bibliography{refs}

\appendix

\section{Searching for a Misaligner}
\label{sec:misaligner-code-description}

\begin{sloppypar}
    In this section, we give a high-level description of our program used to discover the \((320, 676, 8, 0.1625)\)-misaligner from Lemma~\ref{lemma:misaligner-exists}. The naive way of building the collection of strings or codewords would involve starting with three randomly generated codewords satisfying Properties \ref{prop:rate-of-wildcards} and \ref{prop:short-intervals}. Then, at each iteration, a candidate string satisfying Property \ref{prop:rate-of-wildcards} would be generated and tested against the remaining properties. Note that Property~\ref{prop:short-intervals} is easy to check and can be done efficiently. Property~\ref{prop: boundary-blocks} can also be checked naively. Even Property~\ref{prop: full-block-distance-2} can be afforded to be checked naively because we only need to take every pair of strings and compare them with the new string. The most costly property is \ref{prop: full-block-distance-1}, where we need to test the new string $s$ exhaustively against every triple of strings $c_1,c_2,c_3$ already in the collection -- verifying that the edit distance between $s$ and every substring of $c_1 \circ c_2 \circ c_3$ (where a prefix of $s$ aligns with a suffix of $c_1$, a middle portion of $s$ with $c_2$, and a suffix of $s$ with a prefix of $c_3$) meets the required bound. This approach becomes prohibitively slow, especially because Property \ref{prop: full-block-distance-1} requires checking many such alignments. To address this, we need to adopt a more efficient strategy.
\end{sloppypar}


In contrast to the naive approach, our scheme in \texttt{ced-320.cpp} builds a collection that is guaranteed to satisfy Property \ref{prop: full-block-distance-1} (while ignoring the other properties). Instead of testing a new string $s$ against every possible triple $(c_1, c_2, c_3)$, the program first performs an extensive preprocessing phase. It generates a large number of random strings and computes their pairwise edit distances---focusing separately on the prefix of one string with every other entire string, and the suffix of one string with every other entire string, and also for every possible middle portion of one string with every other entire string. These computations are accumulated into statistical tables. \texttt{Pstat} stores the number of pairs of strings with a certain edit distance when the prefix of one is compared to the other string. Similarly, \texttt{Sstat} and \texttt{Mstat} record the corresponding statistics for suffixes and middle sections, respectively. Then, for each possible length of suffix, prefix, or middle section, thresholds are computed as the typical values derived from the stored statistics. By determining the edit distance thresholds (stored in arrays like \texttt{Pthr}, \texttt{Sthr}, and \texttt{Mthr}) corresponding to a very small violation rate (controlled by the parameter \texttt{RANKFRAC}), the program creates a set of ``filters'' that a candidate string must pass.

Now, to build the collection, every new candidate string $c$ is checked again against every string $c'$ already in the collection: the edit distance of every prefix of $c$ with the entire $c'$, the edit distance of every suffix of $c$ with the entire $c'$, and the edit distance of every possible middle substring of $c$ with the entire $c'$ are verified against the corresponding thresholds. 
We also check the thresholds for $c$ and $c'$ interchanged. If any of them is violated, then $c$ isn't added to the collection. 
In this way, for every string $c$ in the final collection, every other string is acceptable when aligning any prefix, suffix, or middle portion of $c$ with it. 
This precisely captures condition \ref{prop: full-block-distance-1}, which requires that when a string $c$ is aligned against the concatenation of three other strings $c_1, c_2, c_3$ (with a prefix of $c$ aligned to $c_1$, 
a middle portion to $c_2$, and a suffix to $c_3$), all such alignments satisfy the specified edit distance thresholds.

This preprocessing-based filtering dramatically reduces the computational load and allows \texttt{ced-320.cpp} to efficiently generate an initial collection where every string satisfies Property~\ref{prop: full-block-distance-1}. Note that the most computationally heavy step is to check Property~\ref{prop: full-block-distance-1}, so we optimized this step by performing extensive preprocessing to compute statistical thresholds that serve as fast filters. This optimization avoids the need for an exhaustive check of every candidate string against every possible triple of existing strings. The rest of the properties are checked naively using the code in \texttt{ted-320.cpp}.
This code checks whether each of the properties is satisfied, and if it is not satisfied it suggests strings that should be removed from the collection to satisfy a given property.

\section{Explicit Bounds for Synchronization Strings on Alphabets of Size Four}
\label{sec:sync-4}

In this section, we prove Corollary~\ref{cor:sync-4}, which states that there exists a $0.999606$-synchronization string over an alphabet of size four. To do so, we first need the following definition.
\begin{defi}
    A string $w$ of length $n$ is called an $\varepsilon$-weak synchronization string if for all $1 \leq i < j < k \leq n$ such that $k-i\geq \frac{1}{1-\varepsilon}$, we have $\ed\left(w_{|[i, j-1]}, w_{|[j, k-1]}\right) \geq (1-\varepsilon)(k-i)$.
\end{defi}
\begin{lemma}
\label{lemma:weak-sync}
    There exists an infinite 0.995268-weak synchronization string over the binary alphabet.
\end{lemma}
\begin{proof}
    By Theorem 6.3 in~\cite{cheng2018synchronization}, if $\varepsilon'$-synchronization strings exist over an alphabet of $2^k$, then $\varepsilon$-weak synchronization strings exist over a binary alphabet, where $\varepsilon := 1-\frac{1-\varepsilon'}{18k}$. Furthermore, every $\frac{\varepsilon}{2}$-locally self-matching string is an $\varepsilon$-synchronization string (Theorem 6.4 in~\cite{haeupler2017synchronization}). Thus by Theorem~\ref{thm:lll-string} and finding the minimum value of the expression \[1-\frac{1-\frac{\varepsilon'}{2}}{18\log_{2}\left(\frac{4e^2}{\varepsilon'^2}\cdot \left(1+4\sqrt[4]{\frac{\varepsilon'}{2}}\right)\right)}\] as $\varepsilon'$ ranges in $\left(0, \frac{1}{2}\right]$, one gets the desired conclusion. 
\end{proof}
We also make use of the following lemma.
\begin{lemma}
\label{lemma:weak-sync-to-sync-4}
    Given an infinite $\varepsilon$-weak synchronization string over a binary alphabet, one can construct an infinite $\frac{\varepsilon+11}{12}$-synchronization string over an alphabet of size four.
\end{lemma}
\begin{proof}
    This follows directly from Theorem 6.4 in~\cite{cheng2018synchronization}.
\end{proof}
Now Lemmas~\ref{lemma:weak-sync} and \ref{lemma:weak-sync-to-sync-4} together imply Corollary~\ref{cor:sync-4}.

\end{document}